\documentclass[conference]{IEEEtran}
\usepackage{amsfonts}
\usepackage{amssymb}
\usepackage{amsmath}
\usepackage{verbatim}
\usepackage{setspace}

\newtheorem{lemma}{Lemma}
\newtheorem{proposition}{Proposition}
\newtheorem{corollary}{Corollary}

\newtheorem{remark}{Remark}

\usepackage{cite}

\ifCLASSINFOpdf
   \usepackage[pdftex]{graphicx}
\else
   \usepackage[dvips]{graphicx}
\fi
\newcounter{mytempeqncnt}

\def\bb0{{\mathbb{0}}}

\def\bb{{\mathbf{b}}}

\def\bg{{\mathbf{g}}}
\def\bh{{\mathbf{h}}}

\def\bw{{\mathbf{w}}}

\def\b0{{\mathbf{0}}}

\def\bD{{\mathbf{D}}}
\def\bE{{\mathbf{E}}}

\def\bG{{\mathbf{G}}}
\def\bH{{\mathbf{H}}}
\def\bI{{\mathbf{I}}}

\def\bN{{\mathbf{N}}}

\def\bY{{\mathbf{Y}}}

\def\sf0{{\mathsf{0}}}

\usepackage{algorithm}
\usepackage{algpseudocode}
\usepackage{bbm}
\usepackage{epstopdf}
\usepackage{balance}
\IEEEoverridecommandlockouts

\begin{document}
%
\title{Optimization of Power Transfer Efficiency and Energy Efficiency for Wireless-Powered Systems with Massive MIMO}



\author{Talha Ahmed Khan, Ali Yazdan, and Robert W. Heath Jr. \thanks{T. A. Khan and R. W. Heath Jr. are with the Department of Electrical and Computer Engineering at The University of Texas at Austin, 2501 Speedway, Austin, TX 78712, USA. A. Yazdan is with Facebook Inc., 1 Hacker Way, Menlo Park, CA 94025, USA (Email: talhakhan@utexas.edu, rheath@utexas.edu, ayp@fb.com).} 
\thanks{T. A. Khan and R. W. Heath Jr. were supported in part by grant W911NF-14-1-0460 from the Army Research Office and a gift from Mitsubishi Electric Research Labs.}
\thanks{This work was presented in part at IEEE 2017 VTC-Spring \cite{VTC17spr}.}}
\maketitle
\begin{abstract}
Massive MIMO is attractive for wireless information and energy transfer due to its ability to focus energy towards desired spatial locations. In this paper, the overall power transfer efficiency (PTE) and the energy efficiency (EE) of a wirelessly powered massive MIMO system is investigated where a multi-antenna base-station (BS) uses wireless energy transfer to charge single-antenna energy harvesting users on the downlink. The users may exploit the harvested energy to transmit information to the BS on the uplink. The overall system performance is analyzed while accounting for the nonlinear nature of practical energy harvesters. First, for wireless energy transfer, the PTE is characterized using a scalable model for the BS circuit power consumption. The PTE-optimal number of BS antennas and users are derived. Then, for wireless energy and information transfer, the EE performance is characterized. The EE-optimal BS transmit power is derived in terms of the key system parameters such as the number of BS 
antennas and the number of users. As the number of antennas becomes large, increasing the transmit power improves the energy efficiency for moderate to large number of antennas. Simulation results suggest that it is energy efficient to operate the system in the massive antenna regime. 
\end{abstract}
%
\IEEEpeerreviewmaketitle

\section{Introduction}
Massive multiple-input multiple-output (MIMO) architecture is a key technology for enabling future 5G networks \cite{heath2014five,marzetta2010massive,khan2017mag}.
Due to its ability to beam energy towards desired spatial regions, massive MIMO is attractive for wireless energy transfer\cite{kayshap2015massivewet,Bi2015wpcsurvey,ref2}. This could enable a wirelessly powered operation for the massive number of RF (radio frequency) energy harvesting devices in future paradigms such as the Internet of Things (IoT) \cite{EnergyHarvestWirelessCommSurvey2015,GollakotaRF,IoT2014}. An
RF or wireless energy harvesting device extracts energy from the incident RF signals. Such wirelessly powered systems are becoming more feasible due to the reduction in the power consumption requirements of devices and the advancement in energy harvesting technologies\cite{GollakotaRF,RFsurveyLondon,valenta2014,talla2015powering}. 

\subsection{Motivation and Related Work} Energy efficiency (EE) has been a key consideration in the system-level analyses of massive MIMO systems\cite{Ngo2013EE,ref3,liu2015massive}. It is often characterized by the ratio of the achievable data rate (bits/sec) and the total power consumption (watts).
While deploying more antennas at the base-station (BS) boosts the data rate, the additional antenna circuitry leads to increased power consumption. This motivates the need for an energy efficient system design. 
In \cite{Ngo2013EE}, the energy efficiency of a massive MIMO system was analyzed while ignoring the circuit power consumption. It was shown that the energy efficiency improves as more antennas are added to the BS. 
Unlike \cite{Ngo2013EE} which considered the transmit power consumption only, the work in \cite{ref3,liu2015massive} investigated the energy efficiency of a massive MIMO system while accounting for the BS circuit power consumption. In \cite{ref3}, it was shown that the transmit power should be increased with the number of antennas for an energy efficient system operation. Moreover, the energy efficiency eventually vanishes in the large-antenna regime. In \cite{liu2015massive}, the downlink energy efficiency of a massive MIMO system was analyzed for a spatially correlated channel model. It was shown that the optimal transmit power is independent of the number of antennas in pilot-contaminated systems. None of this work \cite{Ngo2013EE,ref3,liu2015massive} considered wireless energy and information transfer. 

The energy efficiency and power transfer efficiency (PTE) of RF-powered systems have also been investigated \cite{ref1,ee2017massive,pte2015}.
In \cite{ref1}, a single-user wireless information and power transfer system with a massive antenna array was considered. By jointly optimizing the power transfer duration and the transmit power, an energy efficient resource allocation strategy was proposed under a delay constraint.
In \cite{ee2017massive}, the energy efficiency of a wirelessly powered multi-user massive MIMO system with imperfect channel knowledge was considered. A resource allocation algorithm was designed for optimizing the system parameters such as the number of antennas and power transfer duration.   
In \cite{pte2015}, the power transfer efficiency of a multi-user wireless energy transfer system was investigated. It was shown that the power transfer efficiency can be improved with opportunistic scheduling as the number of users is increased. In other related work, the throughput optimization of massive MIMO wireless information and power transfer systems has also been studied\cite{ref2}. In \cite{ref2}, a throughput-optimal resource allocation policy was proposed for the large-antenna regime. A key limitation of  \cite{ref1,ee2017massive,pte2015,ref2} lies in 
assuming an ideal energy harvesting model and/or a fixed BS power consumption model, which may lead to misleading conclusions in practice.

\subsection{Contributions} In this paper, we characterize the power transfer efficiency and the energy efficiency of a massive MIMO wireless energy and information transfer system using a scalable power consumption model.
Using a piecewise linear energy harvesting model, we derive the average harvested power at a user while accounting for imperfect channel knowledge.
We first focus on wireless energy transfer and analyze the system-level power transfer efficiency. We characterize the optimal number of BS antennas
and users that maximize the power transfer efficiency. We find that the optimal design is guided by the BS power consumption as well as the energy harvesting parameters.
We then consider the case of wireless energy and information transfer where the users exploit the harvested energy to communicate with the BS. We analytically characterize the optimal BS transmit power for an energy efficient system operation. Moreover, we examine the interplay between energy efficiency and the key system parameters. Numerical results suggest that both power transfer efficiency and energy efficiency benefit from operating the system in the massive antenna regime.

\begin{figure}[t]
\centering
{\resizebox{\columnwidth}{!} 
{\includegraphics{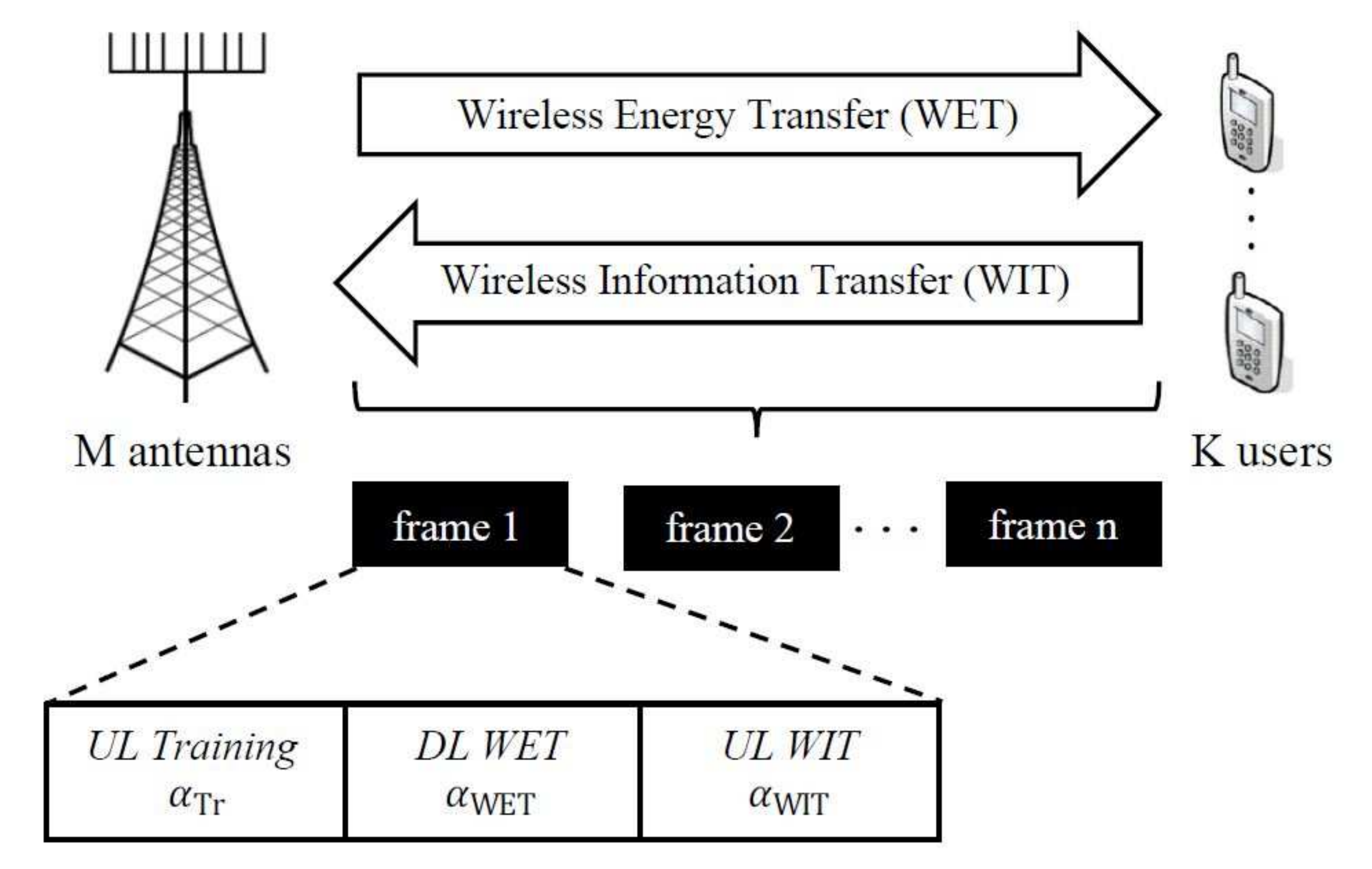}}}  
\caption{System Model.} 
   \label{fig:sys}
\end{figure} 

This paper differs from other related work in several important ways.
First, most prior work investigating the power transfer efficiency or energy efficiency of wireless-powered systems either ignores the BS circuit power consumption or treats it as a fixed component
\cite{pte2015,ref2,ref1}.
This could be misleading since the total power consumption varies with various system parameters such as the number of antennas, the number of users, and the choice of the transmit/receive filters. 
We address this concern by using a scalable power consumption model.
Second, the existing analyses \cite{ref1,ee2017massive,pte2015,ref2} typically consider an ideal energy harvester (EH), where the output power is a scalar multiple of the input. This affords analytical simplicity but it could be misleading in practice. This is because the output power of a practical energy harvester is a nonlinear function of the input. More recently, nonlinear energy harvesting models have been proposed to address this concern \cite{nonlinearEH2015,nonlinearEH2017}. 
In \cite{nonlinearEH2015}, a logistic function was considered for modeling the harvester, and a resource allocation algorithm was designed to maximize the harvested power. In \cite{nonlinearEH2017}, a similar model was used while studying the throughput maximization problem in a multi-user MIMO wireless-powered communication system with separate stations for energy transfer and data reception. To determine the model parameters, the model from \cite{nonlinearEH2015,nonlinearEH2017} relies on curve fitting using measurement data for the energy harvesting circuit under consideration. 
In this paper, we use a piecewise linear model for the energy harvester which abstracts the input-output relationship via activation and saturation thresholds. This captures the key limitations of a practical energy harvester while keeping the analysis tractable. With this motivation, we investigate the power transfer efficiency and energy efficiency of a remotely-powered system using realistic models for energy harvesting and power consumption. This paper is an extension of our previous conference/magazine article \cite{VTC17spr,khan2017mag} where only the energy efficiency of a similar setup was analyzed for the case of ideal energy harvesters. Unlike our previous work, this paper also provides an analytical treatment of power transfer efficiency while considering a realistic model for energy harvesting and power consumption.   

The rest of the paper is organized as follows. Section \ref{Sec: sytem model} describes the system model. Section \ref{SecAvg} derives the average received power at a user which sets the stage for the ensuing analysis. Section \ref{Sec: wet} characterizes the performance of wireless energy transfer in terms of the average harvested power and the power transfer efficiency. Section \ref{sec WIPT} analyzes the performance of wireless energy and information transfer in terms of the data rate and the energy efficiency. Finally, Section \ref{secConc} concludes the paper.

\section{System Model}  \label{Sec: sytem model}
\subsubsection*{Channel Model} \label{Sec: channel model}
We consider a wireless-powered communications system consisting of a BS with $M$ antennas and $K$ single-antenna users. 
We assume each user is equipped with an RF energy harvesting module. 
The BS charges the users on the downlink (DL) and the users exploit the harvested energy to communicate with the BS on the uplink (UL).
 We assume a TDD (time division duplex) mode of operation consisting of a downlink wireless energy transfer (WET) phase, an uplink wireless information transfer (WIT) phase, and an uplink training phase (see Fig. \ref{fig:sys}) \cite{ref2}. We assume the BS learns the uplink channels for each user in the uplink training phase, and uses channel reciprocity to learn the downlink channels. It uses the estimated channel for decoding information on the uplink, and beamforming energy on the downlink. The energy harvesting users, however, are not assumed to have any channel knowledge. 
 
We define $S=T_{\rm{c}}B_{\rm{c}}$ as the length of the coherence block or the frame size, where $T_{\rm{c}}$ and $B_{\rm{c}}$ denote the coherence time and the coherence bandwidth of the wireless channel. The frame is divided into three phases such that a fraction $\alpha_{\textrm{Tr}}\in(0,1)$ is reserved for uplink training, a fraction $\alpha_{\textrm{WET}}\in(0,1)$ for wireless energy transfer, and a fraction $\alpha_{\textrm{WIT}}\in(0,1)$ for wireless information transfer. Moreover, we assume that $\alpha_{\textrm{WET}}+\alpha_{\textrm{WIT}}+\alpha_{\textrm{Tr}}=1$, and set $\alpha_{\textrm{Tr}}=\frac{\tau}{S}$ (where $K\leq\tau<S$) proportional to the number of users. 
We let ${\textbf{h}}_i=\left[h_{i1},\cdots,h_{iM}\right]^T\in\mathbb{C}^{M\times1}$ be the uplink channel from a user $i$ to the BS, where $i\in\{1,\cdots,K\}$. We assume a rich scattering environment with sufficiently spaced antennas such that $h_{ij}$~${\raise.17ex\hbox{$\scriptstyle\mathtt{\sim}$}}$
 $\mathcal{CN}(0,1)$ is a zero-mean complex Gaussian random variable with unit variance, which is independent across $i$ and $j$. While our model assumes Rayleigh fading, it will also be useful to study other fading environments such as Ricean fading in future.
  We use $\beta_{ij}=Cd_{ij}^{-\alpha}$ to model the large-scale gain for the link from user $i$ to the $j$th BS antenna, where $d_{ij}$ denotes the link distance, $\alpha>2$ is the path loss exponent, and $C>0$ is the path loss intercept. We define $\beta_{i}=\sum_{j=1}^{M}\frac{\beta_{ij}}{M}$ as the average large scale gain for user $i$, and treat $d_i=\left(C/\beta_i\right)^{1/\alpha}$ as the corresponding link distance for user $i$.
 We assume that the users are uniformly distributed around the BS in an annulus with inner radius $r_{\min}$ and outer radius $r_{\max}$ such that $\forall~i\in\mathcal{I}_K$,
 the probability density function $f_{d_i}(r)=2r/(r_{\max}^2-r_{\min}^2)$ for
  $r_{\min}\leq r\leq r_{\max}$ and $f_{d_i}(r)=0$ otherwise. By averaging over the user locations, it follows that $\mathbb{E}[{d_i}^{-\alpha}]=\int\limits_{r_{\min}}^{r_{\max}}r^{-\alpha}f_{d_i}(r) \text{d}r
  =\frac{{r_{\max}}^{2-\alpha}-{r_{\min}}^{2-\alpha}}{(1-0.5\alpha)\left({r_{\max}}^{2}-{r_{\min}}^{2}\right)}$.
 We further define $\bH=\begin{bmatrix}\bh_1, \bh_2,\cdots,\bh_K\end{bmatrix}\in\mathbb{C}^{M\times K}$ and 
 $\bG=\begin{bmatrix}\bg_1, \bg_2,\cdots,\bg_K\end{bmatrix}\in\mathbb{C}^{M\times K}$ such that $\bG=\bH\bD^{1/2}$, where $\bD$ is a diagonal matrix with $(\beta_1,\cdots,\beta_K)$ as the entries of the main diagonal. 
We let $\hat{\bG}=\begin{bmatrix}\hat{\bg_1}, \hat{\bg_2},\cdots,\hat{\bg}_K\end{bmatrix}$ denote the channel estimate of $\bG$ at the BS. 
 For the uplink, we denote the average transmit power (in watts) at a user by $P_{\rm{ul}}=\alpha_{\rm{WIT}}B{p_{\rm{ul}}}$, where $p_{\rm{ul}}$ gives the average transmit symbol energy (in joules/symbol), while $B$ denotes the system bandwidth. 
 The user draws the uplink transmit power from the energy it harvests in the downlink. Similarly, for the downlink, 
 $P_{\rm{dl}}={\alpha_{\rm{WET}}B}{p_{\rm{dl}}}$ denotes the average BS transmit power (in watts), and $p_{\rm{dl}}$ (joules/symbol) gives the average downlink transmit energy in one symbol duration. We clarify that the downlink symbol is only an energy-bearing symbol which carries no information.
 We further note that the transmit signal waveform for wireless power transfer presents another design degree of freedom \cite{bruno2017waveform}. We do not consider waveform optimization in this paper.

\subsubsection*{Energy Harvesting Model}We assume that each user is equipped with an RF energy harvesting module with a sufficiently large battery. To simplify the analysis, prior work mostly assumes an \textit{ideal} energy harvester where the harvested energy scales linearly with the input power. In practice, however, an energy harvester is a nonlinear device with a small operating range, which may lead to vastly different performance trends compared to the ideal case\cite{nonlinearEH2015}. For example, the incident energy should be sufficiently high to activate the harvester; not all the incident energy can be harvested; and the harvester output eventually saturates beyond a certain input power. 
We, therefore, strengthen the analysis by parameterizing the harvester operation using $\{\theta_\textrm{act},\theta_\textrm{sat},\eta_\textrm{EH}\}$: $\theta_\textrm{act}$ is the harvester activation threshold (watts), $\theta_\textrm{sat}$ is the harvester saturation threshold (watts), and $\eta_\textrm{EH}\in(0,1]$ is the rectifier efficiency. An ideal energy harvester has $\theta_\textrm{act}=0$ and $\theta_\textrm{sat}=\infty$. We will often call a harvester in the active mode to be in the non-saturated mode.
\subsubsection*{Notation}
For a positive integer $K$, we define the index set $\mathcal{I}_K=\{1,\cdots,K\}$. We use the superscripts $*$ and $\rm{H}$ to denote conjugate and conjugate transpose of a matrix. We use $\lceil x \rceil$ and $\lfloor x \rfloor$ to denote the integer ceiling or the integer floor of a real number $x$.
\section{Average Received Energy}\label{SecAvg}
In this section, we analytically characterize the average received (incident) energy at the users assuming perfect and imperfect channel state information (CSI) at the BS. The corresponding harvested energy is characterized in the next section. 
\subsection{Average Received Energy: Perfect CSI}
We assume that the BS transmits with the average transmit energy $p_{\rm{dl}}$ (in joules/symbol) in the downlink. 
The BS uses a weighted sum of conjugate beamformers for each user in the downlink, since it has been shown to be asymptotically optimal for wireless energy transfer\cite{zhang2013mimo}.
The precoder $\bw_{\rm{dl}}=\sum_{i=1}^{K} \sqrt{\zeta_i}\frac{\bw_i}{\|\bw_i\|}$ where $\bw_i=\hat{\bg}_i$, and $\zeta_i\in(0,1)~\forall~i$ such that $\sum_{i=1}^{K}\zeta_i=1$. 
Assuming the BS transmits a signal $s$ with $\mathbb{E}[|s|^2]=p_{\rm{dl}}$, the signal $y_i$ received at user $i$ can be expressed as
\begin{align}\label{eq:signal main}
y_i&=\bg_i^{\rm{H}}\bw_{\rm{dl}}s+n_i=\sqrt{\zeta_i}{\bg_i}^{\rm{H}}\frac{\hat{\bg}_i}{\|\hat{\bg}_i\|}s+
\sum_{j\neq i}^{K}\sqrt{\zeta_j}\bg_i^{\rm{H}}\frac{\hat{\bg}_j}{\|\hat{\bg}_j\|}s+n_i,
\end{align}
where $n_i$ is the receiver noise. A user harvests energy from the beam directed towards it, as well as from those directed towards other users. Assuming perfect channel knowledge at the BS such that $\hat{\bg_i}=\bg_i\,\forall\,i\in\mathcal{I}_K$, (\ref{eq:signal main}) simplifies to   
\begin{align}\label{eq:signal}
y_i&=\sqrt{\zeta_i}{\|\bg_i\|}s+
\sum_{j\neq i}^{K}\sqrt{\zeta_j}\bg_i^{\rm{H}}\frac{\bg_j}{\|\bg_j\|}s+n_i.
\end{align}
The contribution from the noise term is usually negligible and is therefore ignored. This results in the following analytical expression for the average received energy $\bar{\gamma}_i=\alpha_{\rm{WET}}\,\mathbb{E}\left[|y_i|^2\right]$ at a user $i$.
\begin{lemma}\label{lemma:avgrcvperfectcsi}
When a BS with $M$ antennas serves $K$ single-antenna energy harvesting users, the average received energy $\bar{\gamma}_i$ (in joules/symbol) at a user $i$, assuming perfect channel knowledge at the BS, is given by 
\begin{align}\label{eq:avg rcv perfect lem1}
	\bar{\gamma}_i= \alpha_{\rm{WET}}\,p_{\rm{dl}}\,\beta_i\left(\zeta_iM+ \left(1-\zeta_i\right)\right)
\end{align}
where $\alpha_{\rm{WET}}$ denotes the fraction reserved for downlink energy transfer, $p_{\rm{dl}}$ gives the transmit symbol energy (joules/symbol), and $\beta_i$ gives the large-scale channel gain. 
\end{lemma}
\begin{proof}
See Appendix.
\end{proof}
The parameter $\alpha_{\rm{WET}}$ captures the fact that the users receive energy for a fraction $\alpha_{\rm{WET}}$ of the frame. The average received energy during the entire frame is given by $S\bar{\gamma}_i$.
We note that the average received power $B\bar{\gamma}_i$ increases with an increase in the number of BS antennas $M$. Its dependency on the number of users is captured by the energy allocation parameter $\zeta_i$, which tends to decrease as more users are added to the system. Moreover, the term $\zeta_iM$ is due to the BS transmission intended for user $i$, while $1-\zeta_i$ results from the transmissions intended for other users. 
\begin{corollary}
The average received energy $\bar{\gamma}_i\leq \alpha_{\rm{WET}}\,p_{\rm{dl}}\,\beta_i M$, which holds with equality for the single-user scenario where $\zeta_i=1$.
\end{corollary}
\begin{corollary} 
\label{cor: rcv pwr perfect}
Under an equal transmit energy allocation at the BS, i.e., $\zeta_i=\frac{1}{K}\,\forall\,i\in \mathcal{I}_K$, the average received energy is given by
\begin{align}\label{eq:avg rcv cor 1}
	\bar{\gamma}_i= \alpha_{\rm{WET}}\,p_{\rm{dl}}\,\beta_i\left(1+\frac{M-1}{K}\right).
\end{align}
\begin{proof}
This follows by plugging $\zeta_i=\frac{1}{K}$ in (\ref{eq:avg rcv perfect lem1}).
\end{proof}
\end{corollary}

\begin{corollary}\label{cor:m/k ratio}
The average received energy converges to  $\lim\limits_{M,K\rightarrow\infty}\bar{\gamma}_i=\alpha_{\rm{WET}}\,p_{\rm{dl}}\,\beta_i\,\left(1+r\right)$ as both $M$ and $K$ grow large with $\frac{M}{K}=r>1$ held constant.
\begin{proof}
The result follows directly from Corollary \ref{cor: rcv pwr perfect}.
\end{proof}
\end{corollary}
Therefore, increasing the ratio $r$ helps improve the average received energy at the users. This is because adding more antennas boosts the beamforming gain, and serving fewer users increases the per user energy allocation at the BS.
\subsection{Average Received Energy: Imperfect CSI}
We now characterize the average received energy while incorporating the channel estimation errors in the analysis. Imperfect channel estimation causes a reduction in the amount of energy reaching the harvesters. 
We recall that a fraction ${\alpha_{\rm{Tr}}}$ of the frame is reserved for uplink training.
We assume that the $K$ users simultaneously transmit their training signals consisting of $\tau$ symbols, where $\tau\geq K$ and ${\alpha_{\rm{Tr}}}=\frac{\tau}{S}$. We define a $\tau\times K$ matrix $\bf{\Phi}$ where the $i$th column contains the training sequence of user $i$. We assume that the users transmit orthogonal training sequences such that $\bf{\Phi^{\rm{H}}\Phi}={\bf{I}}_{\rm{K}}$. Let us define a diagonal matrix $\bf{\Delta}$ with $\{\tau p_{{\rm{Tr}},1},\cdots,\tau p_{{\rm{Tr}},K}\}$ as its diagonal entries, where $p_{{\rm{Tr}},i}$ denotes the training symbol energy of user $i$.
The signal received at the BS during the training phase can be expressed as
\begin{align}\label{pilot}
\bf{ Y_{\rm{Tr}} = G (\Phi \Delta^{\frac{1}{2}})^{\rm{T}} + N} 
\end{align} 
where the $M\times\tau$ matrix $\bN$ denoting the BS thermal noise consists of IID Gaussian entries with mean zero and variance $\sigma^2$. 
For a user $i$, we define $\xi_i\in (0,1)$ as the fraction of the total energy harvested $\bar{\delta}_i^{\rm{I}}S$ (treated in Lemma \ref{lem:ls harv}) in a frame that is reserved for uplink pilot transmission.
Therefore, 
$\tau p_{{\rm{Tr}},i}=\eta^{\rm{EH}}_{\rm{PA}}\xi_i\bar{\delta}_i^{\rm{I}}S$, where $\eta^{\rm{EH}}_{\rm{PA}}\in(0,1)$ is the power amplifier (PA) efficiency at the energy harvesting user.
\subsubsection{LS Channel Estimation}\label{SecLSest}
We first consider the case where the BS estimates the UL channel from the $K$ EHs using linear least squares (LS) approach. The resulting channel estimate ${\hat{\bG}}_{\rm{LS}}$ is given by
\begin{align}\label{eq:ls estimate}
{\hat{\bG}}_{\rm{LS}}=\bY_{\rm{Tr}}\bf{\Phi}^{*}{\Delta}^{\rm{-\frac{1}{2}}}=G+\bN\bf{\Phi}^{*}{\Delta}^{\rm{-\frac{1}{2}}}.
\end{align}
The corresponding estimation error matrix $\bE_{\rm{LS}}={\hat{\bG}}_{\rm{LS}}-\bG_{\rm{LS}}=\bN\bf{\Phi}^{*}{\Delta}^{\rm{-\frac{1}{2}}}$ consists of independent Gaussian entries $e^{\rm{LS}}_{ij} \left(i\in\mathcal{I}_M, j\in\mathcal{I}_K\right)$ with mean zero and variance $\frac{\sigma^2}{\tau p_{{\rm{Tr}},i}}$.
The following expression characterizes the mean incident power at a user $i$. 

\begin{lemma}\label{lem:ls}
When the BS designs the downlink energy beamformer based on the LS channel estimate, the average received energy $\bar{\gamma}_i^{\rm{LS}}$ (in joules/symbol) at a user $i$ is given by 
\begin{align}\label{eq:ls rcv}
\bar{\gamma}_{i}^{\rm{LS}}&=
\begin{cases}
& {\psi}_{i}^{\rm{LS,act}},\quad  \frac{\theta_{{{\rm{act}}}}}{B}\leq{\psi}_{i}^{\rm{LS,act}}<\frac{\theta_{{{\rm{sat}}}}}{B}\\
& {\psi}_{i}^{\rm{LS,sat}},\quad\qquad\,\,\,\, {\psi}_{i}^{\rm{LS,act}}\geq\frac{\theta_{{{\rm{sat}}}}}{B}\\
\end{cases}
\end{align}
where 
\begin{align}\label{eq: rcv ls act}
&{\psi}_{i}^{\rm{LS,act}}=
\nonumber\\
&\frac{A_1 M+A_2-A_3+\sqrt{\left(A_1 M+A_2-A_3\right)^2+4\left(A_1+A_2\right)A_3}}{2},
\end{align}
\begin{align}\label{eq:rcv ls sat}
{\psi}_{i}^{\rm{LS,sat}}&=
A_1M\left(1-\frac{M-1}{M}
\frac{1}{1+\frac{\theta_{\rm{sat}}}{BA_3}}\right)+A_2,
\end{align}
\begin{align}
A_1=\alpha_{\rm{WET}}p_{\rm{dl}}\beta_i\zeta_i,
\end{align}
\begin{align}
A_2=\alpha_{\rm{WET}}p_{\rm{dl}}\beta_i\left(1-\zeta_i\right),
\end{align}
and
\begin{align}\label{eq: xyz}
A_3=\frac{\sigma^2}{\xi_i\beta_i\eta^{\rm{EH}}_{\rm{PA}}\eta_{\rm{EH}}S}.
\end{align}
\begin{proof}
See Appendix.  
\end{proof}
\end{lemma}
We can interpret (\ref{eq:ls rcv}) as follows. The average received energy is given by the expression ${\psi}_{i}^{\rm{LS,act}}$ as long as the corresponding incident power falls within the linear range $\left[\theta_{{{\rm{act}}}},\theta_{{\rm{sat}}}\right)$ of the harvester. It is given by the expression ${\psi}_{i}^{\rm{LS,sat}}$ when the corresponding received power exceeds the saturation threshold of the harvester. 
When the incident power level is within the linear regime of the harvester, the harvested power increases with the incident power. As it exceeds the saturation threshold, however, the harvested power remains the same regardless of the incident power. This explains why different analytical expressions are required to characterize the incident energy. 
The following remark explains why the incident power, which is not the same as the harvested power, also depends on the harvesting parameters. 

\begin{corollary}\label{cor: ideal harv incident}
For an ideal energy harvester with activation threshold $\theta_{{\rm{act}}}\rightarrow 0$ and saturation threshold $\theta_{{\rm{sat}}}\rightarrow\infty$, the average received energy simplifies to $\bar{\gamma}_i^{\rm{LS}}={\psi}_i^{\rm{LS,act}}$. 
\end{corollary}
\begin{remark}
We note from (\ref{eq:ls rcv})--(\ref{eq:rcv ls sat}) and (\ref{eq: xyz}) that the incident energy at a user also depends on the energy harvesting parameters. This is because the channel estimation error is a function of the uplink transmit power, which is drawn from the energy harvested in the previous frames. 
This introduces a dependency between the downlink energy beamformer and the energy harvesting parameters, as evident from the analytical expressions in Lemma \ref{lem:ls}.
We further add that the average received energy could be smaller than $\frac{\theta_{{\rm{act}}}}{B}$. Since this amount would be insufficient to activate the harvester, we do not consider this case in Lemma \ref{lem:ls}. In principle, we may characterize this by assuming omnidirectional transmission, since the BS would not have any channel knowledge in the absence of uplink training. 
\end{remark}

\begin{remark}
The sum $A_1M+A_2$ in Lemma \ref{lem:ls} equals $\bar{\gamma}_i$, which is the average received energy with perfect CSI. Moreover, the term $A_3$ captures the dependency on the EH parameters and the BS noise. As $\sigma^2\rightarrow 0$, so does the estimation error and we recover the expression for the case with perfect CSI. Similarly, the degradation due to imperfect CSI vanishes as the frame size $S\rightarrow\infty$ and $A_3\rightarrow 0$. This is because the users can afford a larger transmit power during pilot transmission due to an underlying increase in the energy harvested in a frame.
\end{remark}
Finally, we note that the average received energy increases with an increase in the number of BS antennas, the EH conversion efficiency, as well as the PA efficiency at the user. It reduces with an increase in the number of users due to a decrease in the per-user transmit energy allocation at the BS.

\subsubsection{MMSE Channel Estimation}\label{SecMMSEest}
We now consider the case where the BS estimates the uplink channel using (linear) minimum mean squared error (MMSE) estimation. The estimated channel is given by
\begin{align}\label{eq: mmse}
{\hat{\bG}}_{\rm{MMSE}}=\bY_{\rm{Tr}}{\bf{\Phi}}^{*}{\left({\bf{D\Delta}}+\sigma^2{\bI}_{K}\right)^{-1}}
{\bf{\Delta}}^{\rm{\frac{1}{2}}}\bD 
\end{align}
The corresponding estimation error matrix $\bE_{\rm{MMSE}}={\hat{\bG}}_{\rm{MMSE}}-\bG_{\rm{MMSE}}$ consists of entries $e^{\rm{MMSE}}_{ij} \left(i\in\mathcal{I}_M, j\in\mathcal{I}_K\right)$ with mean zero and variance $\frac{\beta_i}{1+\frac{\beta_i\tau p_{{\rm{Tr}},i}}{\sigma^2}}$. This error variance is smaller than that obtained with LS estimation. Moreover, the matrices $\bE_{\rm{MMSE}}$ and ${\hat{\bG}}_{\rm{MMSE}}$ are independent by virtue of the orthogonality principle and the fact that uncorrelated Gaussian random variables are independent. 

\begin{remark}\label{rem: ls=mmse}
We note that the LS channel estimate is a scalar multiple of that obtained with the MMSE approach. Specifically, $\hat{{\bg}}_i^{\rm{LS}}= \left(1+\frac{\sigma^2}{\beta_i\tau p_{{\rm{Tr}},i}}\right)\hat{{\bg}}_i^{\rm{MMSE}}$ for $i\in\mathcal{I}_K$, which follows by simplifying (\ref{eq:ls estimate}) and (\ref{eq: mmse}). In other words, the phase of the estimated channel remains the same with LS and MMSE. 
\end{remark}

The following expression characterizes the mean incident energy at a user $i$. 

\begin{lemma}\label{lem:msse}
When the BS designs the downlink energy beamformer based on the MMSE channel estimate, the average received energy at a user $i$ is given by $\bar{\gamma}_{i}^{\rm{MMSE}}=\bar{\gamma}_{i}^{\rm{LS}}$ where $\bar{\gamma}_{i}^{\rm{LS}}$ follows from
Lemma \ref{lem:ls}.
\begin{proof}
The proof follows from Remark \ref{rem: ls=mmse} and by noting that the beamfomer in (\ref{eq:signal main}) consists of normalized vectors such that $\frac{\hat{\bg}_i^{\rm{LS}}}{\|\hat{\bg}_i^{\rm{LS}}\|}=
\frac{\hat{\bg}_i^{\rm{MMSE}}}{\|\hat{\bg}_i^{\rm{MMSE}}\|}$,  resulting in the same energy.
\end{proof}
\end{lemma}
The average received energy obtained with MMSE estimation is the same as that obtained with the LS approach\footnote{While this equivalence holds in a rich scattering environment modeled by IID Rayleigh fading, this may not be the case in other scenarios such as a Ricean fading environment \cite{kayshap2015massivewet}. 
}. This is because the channel estimates obtained with both approaches differ only by a scaling factor. The downlink beamfomer in (\ref{eq:signal main}) consists of normalized vectors such that $\frac{\hat{\bg}_i^{\rm{LS}}}{\|\hat{\bg}_i^{\rm{LS}}\|}=
\frac{\hat{\bg}_i^{\rm{MMSE}}}{\|\hat{\bg}_i^{\rm{MMSE}}\|}$. This means that the downlink energy beamformer and the received energy is the same with both approaches. Under this equivalence, the LS approach is preferable as it does not require statistical knowledge about the channel covariance matrix.
In the rest of the paper, we will not distinguish between LS or MMSE estimation. We use $\bar{\gamma}_{i}^{\rm{I}}\triangleq\bar{\gamma}_{i}^{\rm{MMSE}}=\bar{\gamma}_{i}^{\rm{LS}}$ to refer to the average received energy with imperfect channel knowledge.

\section{Wireless Energy Transfer}\label{Sec: wet}
In this section, we focus on wireless energy transfer where the BS attempts to charge users, but no information transfer is considered. 
We analyze the average harvested energy and the power transfer efficiency in terms of the system parameters.
Wireless energy and information transfer is treated in Section \ref{sec WIPT}.
\subsection{Average Harvested Energy}
Leveraging the analysis in Section \ref{SecAvg}, 
we provide analytical expressions for the average harvested energy for the case of perfect ($\bar{\delta}_i$) and imperfect channel knowledge ($\bar{\delta}_i^{\rm{I}}$).
\subsubsection{Perfect CSI}
We first consider the case where the BS has perfect channel knowledge. 
We use $\vartheta_{i,s}$ to denote the energy harvested by a user $i$ in slot $s$ $\left(s\in\mathcal{I}_{\lfloor\alpha_{\rm{WET}}S\rfloor}\right)$ during the harvesting phase of the frame.
We let $\bar{\delta}_i^{}=\alpha_{\rm{WET}}\mathbb{E}\left[\vartheta_{i}\right]$ denote the average harvested energy (in an arbitrary slot) at a user $i$, where we have dropped the subscript $s$ as the mean is identical across the slots in the harvesting phase.
Due to the piecewise linear energy harvesting model, the average harvested energy $\bar{\delta}_{i}$ (in joules/symbol) at a user $i$ is given by
\begin{align}\label{eq: instantaneous harvested energy}
\bar{\delta}_i&=\eta_{\rm{EH}}\bar{\gamma}_i{\mathbbm{1}}_{\left[\frac{\theta_{{\rm{act}}}}{B}\leq\bar{\gamma}_i<\frac{\theta_{{\rm{sat}}}}{B}\right]}
+\frac{\eta_{\rm{EH}}\theta_{{\rm{sat}}}}{B} {\mathbbm{1}}_{\left[\bar{\gamma}_i\geq\frac{\theta_{{\rm{sat}}}}{B}\right]}
\end{align}
where $\mathbbm{1}_{\left[\cdot\right]}$ is the indicator function which is 1 when the condition in the parenthesis is true, and zero otherwise.

\begin{lemma}\label{lem:harperfect}
The average harvested energy $\bar{\delta}_i$ at a user $i$ can be expressed as
\begin{equation}\label{eq: average harvested energy}
\bar{\delta}_i=
\begin{cases} 
0, & M<{M_{{\rm{act}},i}}\\
\eta_{\rm{EH}}\bar{\gamma}_i, & M_{{\rm{act}},i}\leq M< M_{{\rm{sat}},i}\\
\frac{\eta_{\rm{EH}}\theta_{{\rm{sat}}}}{B}, & M\geq {M_{{\rm{sat}},i}} 
\end{cases}
\end{equation}
where $M_{{{\rm{act}}},i}=
\left\lceil1+\frac{1}{\zeta_i}\left(\frac{\theta_{{\rm{act}}}}{\beta_iP_{\rm{dl}}}-1\right)\right\rceil$ for $\theta_{{\rm{act}}}\in(0,\infty)$
and
$M_{{{\rm{sat}}},i} =\left\lceil1+\frac{1}{\zeta_i}\left(\frac{\theta_{{\rm{sat}}}}{\beta_iP_{\rm{dl}}}-1\right)\right\rceil$ for $\theta_{{\rm{sat}}}\in(0,\infty)$ give the minimum number of antennas needed to activate or saturate the harvester. 
\end{lemma}
\begin{proof}
The proof follows from invoking Lemma \ref{lemma:avgrcvperfectcsi} and the definition of the average harvested energy in (\ref{eq: instantaneous harvested energy}). Let us consider the first term in (\ref{eq: instantaneous harvested energy}). Using (\ref{eq:avg rcv perfect lem1}), we find an analytical expression for $M_{{{\rm{act}}},i}$ such that the condition $\frac{\theta_{{\rm{act}}}}{B}\leq\bar{\gamma}_i<\frac{\theta_{{\rm{sat}}}}{B}$ of the indicator function is satisfied. Similarly, we derive an expression for $M_{{{\rm{sat}}},i}$ such that the condition $\bar{\gamma}_i\geq\frac{\theta_{{\rm{sat}}}}{B}$ in the indicator function of the second term in (\ref{eq: instantaneous harvested energy}) is satisfied. 
\end{proof}

\begin{corollary}
For an ideal energy harvester with activation threshold $\theta_{{\rm{act}}}\rightarrow 0$ and saturation threshold $\theta_{{\rm{sat}}}\rightarrow\infty$, the average harvested energy simplifies to $\bar{\delta}_i=\eta_{\rm{EH}}\bar{\gamma}_i$ where $\bar{\gamma}_i$ follows from Lemma \ref{lemma:avgrcvperfectcsi}. 
\end{corollary}

The antenna thresholds $\{M_{{\rm{act}},i}, M_{{\rm{sat}},i}\}$ depend on the downlink BS transmit power $\zeta_iP_{\rm{dl}}=\zeta_i\alpha_{\rm{WET}}Bp_{\rm{dl}}$ for user $i$ and the link attenuation $\beta_i$.
Increasing the BS transmit power, serving fewer users, or deploying a harvester with a smaller activation threshold reduces the number of required antennas. 
In other words, a minimum of $M_{{\rm{act}},i}$ antennas are required for a successful wireless energy transfer to a user $i$, and at least $M_{{\rm{sat}},i}$ antennas are required to operate the harvester at its maximum potential. In a multi-user system, the BS should have 
at least $M=\max\limits_{i\in\mathcal{I}_K} M_{{\rm{act}},i}$ antennas to ensure that all users are served. Similarly, having $M=\max\limits_{i\in\mathcal{I}_K} M_{{\rm{sat}},i}$ BS antennas ensures that each user attains the maximum possible harvested energy. 
With fewer than $M=\min\limits_{i\in\mathcal{I}_K} M_{{\rm{act}},i}$ antennas, all the transmitted energy will go to waste as none of the harvesters will be activated. Likewise, having more than $M=\max\limits_{i\in\mathcal{I}_K} M_{{\rm{sat}},i}$ antennas will not further improve the harvested energy since all the users will be in saturated mode. 
This behavior is markedly different from that observed with an (ideal) linear energy harvesting model where the average harvested energy keeps on increasing with $M$.

\subsubsection{Imperfect CSI}
We now characterize the average harvested energy while accounting for the channel estimation errors at the BS. 
Since LS/MMSE estimation results in the same received energy, we denote the harvested energy as $\bar{\delta}_i^{\rm{I}}$, where the superscript $``\rm{I}"$ signifies imperfect channel knowledge.
\begin{lemma}\label{lem:ls harv}
With LS/MMSE channel estimation, the average harvested energy $\bar{\delta}_i^{\rm{I}}$ at a user $i$ can be approximated as
\begin{equation}\label{eq: average harvested energy ls}
\bar{\delta}_i^{\rm{I}}=
\begin{cases} 
0, & M<{M_{{\rm{act}},i}^{\rm{I}}}\\
\eta_{\rm{EH}}\bar{\gamma}_i^{\rm{I}}, & M_{{\rm{act}},i}^{\rm{I}}\leq M< M_{{\rm{sat}},i}^{\rm{I}}\\
\frac{\eta_{\rm{EH}}\theta_{{\rm{sat}}}}{B}, & M\geq {M_{{\rm{sat}},i}^{\rm{I}}} 
\end{cases}
\end{equation}
where 
\begin{align}\label{eq: mact mmse}
{M_{{\rm{act}},i}^{\rm{I}}}
=\min\left\{M:~ \frac{\theta_{\rm{act}}}{B}\leq {\bar{\gamma}}_{i}^{\rm{I}}<\frac{\theta_{\rm{sat}}}{B}\right\}
\end{align}
and
\begin{align}\label{eq: msat mmse}
{M_{{\rm{sat}},i}^{\rm{I}}}
=\min\left\{M:~  {\bar{\gamma}}_{i}^{\rm{I}}\geq\frac{\theta_{\rm{sat}}}{B}\right\}
\end{align}
denote the minimum number of antennas required to activate or saturate the harvesters, while ${\bar{\gamma}}_{i}^{\rm{I}}$ follows from Lemma \ref{lem:ls}.
\begin{proof}
The proof is similar to that of Lemma \ref{lem:harperfect} and follows from invoking Lemma \ref{lem:ls} along with the definition in (\ref{eq: instantaneous harvested energy}) assuming imperfect channel knowledge. Unlike Lemma \ref{lem:harperfect}, however, the analytical expressions for ${M_{{\rm{act}},i}^{\rm{I}}}$ and ${M_{{\rm{sat}},i}^{\rm{I}}}$ are rather unwieldy. Therefore, we express them in terms of ${\bar{\gamma}}_{i}^{\rm{I}}$ for simplicity.
\end{proof}
\end{lemma}
\begin{corollary}
For an ideal energy harvester with activation threshold $\theta_{{\rm{act}}}\rightarrow 0$ and saturation threshold $\theta_{{\rm{sat}}}\rightarrow\infty$, the average harvested energy simplifies to $\bar{\delta}_i^{\rm{I}}=\eta_{\rm{EH}}\bar{\gamma}_i^{\rm{I}}$ where $\bar{\gamma}_i^{\rm{I}}$ follows from Lemma \ref{lem:ls}. 
\end{corollary}

	As compared to the case with perfect CSI, a larger number of BS antennas is required to drive the users into activation or saturation mode. The expressions (\ref{eq: mact mmse}) and (\ref{eq: msat mmse}) characterize the required number of antennas with LS/MMSE channel estimation. 
	
\begin{remark}\label{rem: eff}
Let us consider the quantity $\left(1-\xi_i\right)\bar{\delta}_i^{\rm{I}}$, which represents the \emph{effective} harvested energy at a user $i$. This is because a fraction $\xi_i$ of the harvested energy is reserved for uplink training. On one hand, increasing $\xi_i$ improves the channel estimation accuracy at the BS, thereby enhancing the incident power at the user. On the other hand, this means that a smaller fraction of the harvested energy will be available to the user. Therefore, $\xi_i$ should be tuned so as to maximize the effective harvested energy at each user. We refer the interested readers to \cite{Zhang2015training} for a comprehensive treatment of training design for wireless energy transfer systems.   
\end{remark}

\subsection{Power Transfer Efficiency}\label{sec:PTE}
We define the system-level power transfer efficiency (PTE) as the ratio of the total average power harvested by all users
to the total average BS power consumption. We find the optimal number of antennas and the optimal number of users to maximize the system-level power transfer efficiency. 

\subsubsection{Perfect CSI}
We first consider the case where the BS has perfect channel knowledge.
For ease of exposition, we set $\beta_i$ 
to the average
$\mathbb{E}[\beta_i]=C\mathbb{E}[d_i^{-\alpha}]\triangleq\beta$, where $\mathbb{E}[d_i^{-\alpha}]$ is given in Section \ref{Sec: sytem model}. 
As $\beta_i\triangleq\beta$ such that $\bar{\delta}_i\triangleq\bar{\delta}\,\,\forall\, i\in\mathcal{I}_K$, we can equivalently view this as a 
symmetric setup where the users are located on a circle of radius $\left({C}{\beta}^{-1}\right)^{\frac{1}{\alpha}}$ around the BS.
We model the total BS power consumption as a sum of  $P_{\rm{TX}}$ and $P_{\rm{c}}$, where $P_{\rm{TX}}$ denotes the total average transmit (PA) power consumption (in watts) and $P_{\rm{c}}$ the total average circuit power consumption (in watts) at the BS.
The PTE of overall system can be formulated as
\allowdisplaybreaks{
\begin{align}\label{eq:PTE}
\text{PTE}\left(M,K\right)&=\frac{\sum_{i=1}^{K}B\bar{\delta}_i}{P_{\rm{TX}}+P_{\rm{c}}}\nonumber\\
&=\frac{B{K}\bar{\delta}}{
P_{\rm{TX}}
+P_{\rm{FIX}}+ MP_{\rm{BS}}+P_{\rm{CE}}+P_{\rm{LP}}
}
\end{align}}
\begin{figure*}[!t]
	\normalsize
	\setcounter{mytempeqncnt}{\value{equation}}
	\setcounter{equation}{21}

\begin{align} \label{eq:PTE K condition}
{K}^*_{\rm{PTE}}=
\begin{cases}
{K}^*_{\rm{sat}}, &
K^{*}_{\rm{sat}}>\frac{\left[K_{\rm{act}}^*+M-1\right]\beta P_{\rm{dl}}}{\theta_{\rm{sat}}}
\frac{\left(K^{*}_{\rm{sat}}\right)^2\dot{P}_{\rm{CE}}+ K^{*}_{\rm{sat}}\dot{P}_{\rm{LP}}+P_{\rm{TX}}+P_{\rm{FIX}}+MP_{\rm{BS}}}{\left(K^{*}_{\rm{act}}\right)^2\dot{P}_{\rm{CE}}+ K^{*}_{\rm{act}}\dot{P}_{\rm{LP}}+P_{\rm{TX}}+P_{\rm{FIX}}+MP_{\rm{BS}}}\\
{K}^*_{\rm{act}}, & \text{else}
\end{cases}
\end{align}
	\setcounter{equation}{\value{mytempeqncnt}}
	\hrulefill
	\vspace*{4pt}
\end{figure*}

In particular, $P_{\rm{TX}}=\frac{P_{\rm{dl}}}{\eta^{\rm{BS}}_{\rm{PA}}}=\frac{\alpha_{\rm{WET}}p_{\rm{dl}}B}{\eta^{\rm{BS}}_{\rm{PA}}}$ where
$\eta^{\rm{BS}}_{\rm{PA}}\in(0,1)$ denotes the BS PA efficiency. Note that the uplink transmit power, which is a fraction of the average harvested power, only appears in the numerator (via the expression for the harvested energy $\bar{\delta}$). This is because the energy harvesting users do not have any power source except for the wireless energy delivered by the BS. 
Inspired by \cite{ref3}, we allow the circuit power consumption $P_{{c}}$ to scale with the key parameters such as $M$ and $K$: $P_{\rm{FIX}}$ lumps the fixed power spent on running the BS; $P_{\rm{BS}}$ models the circuit power consumed by an RF chain such that $MP_{\rm{BS}}$ gives the total power consumed by the antenna circuitry. Let us use $\kappa_{\rm{BS}}$ to denote the BS computational efficiency in flops/watt, and recall that there are $\frac{B}{S}$ coherence blocks per second. Then, 
 $P_{\rm{CE}}=M\tilde{P}_{\rm{CE}}=\frac{2MK^2B}{S\kappa_{\rm{BS}}}$ models the power consumed while computing the channel estimates on the uplink during each coherence block (includes the power consumed in multiplying an $M\times K$ received pilot signal with a length $K$ pilot sequence for each of the $K$ users \cite[Appendix C]{boyd2004convex}); $P_{\rm{LP}}=M\tilde{P}_{\rm{LP}}=\left(\frac{3MKB}{S\kappa_{\rm{BS}}}\right)$ accounts for the power consumption due to linear processing at the BS, i.e., for computing the downlink energy beamformer. 
 We note that the computational power consumption is usually negligible compared to the antenna power consumption in the large-antenna regime. 

\subsubsection*{Optimal M} 
We now characterize the number of antennas required to optimize the power transfer efficiency, assuming the other parameters to be fixed. Let us denote this quantity by ${M}^*_{\rm{PTE}}$.
\begin{proposition}\label{lem:pte M}
In a system with $K$ users, the PTE-optimal number of antennas is given by ${M}^*_{\rm{PTE}}=$
\begin{align}\label{eq:pte m}
\begin{cases} 
{M_{{\rm{act}}}^{\rm{}}}=\left\lceil1+K\left(\frac{\theta_{{\rm{act}}}}{\beta P_{\rm{dl}}}-1\right)\right\rceil
, & K\geq 1+
\frac{P_{\rm{TX}}
+P_{\rm{FIX}}}
{P_{\rm{BS}}+{{\tilde{P}_{\rm{CE}}}}+{\tilde{P}_{\rm{LP}}}}\\
M_{{\rm{sat}}}=\left\lceil1+K\left(\frac{\theta_{{\rm{sat}}}}{\beta P_{\rm{dl}}}-1\right)\right\rceil, & K< 1+
\frac{P_{\rm{TX}}
+P_{\rm{FIX}}}
{P_{\rm{BS}}+{{\tilde{P}_{\rm{CE}}}}+{\tilde{P}_{\rm{LP}}}} 
\end{cases}
\end{align}
where ${M_{{\rm{act}}}^{\rm{}}}$ and ${M_{{\rm{sat}}}^{\rm{}}}$ are as defined in Lemma \ref{lem:harperfect}.
\begin{proof}
See Appendix.
\end{proof}
\end{proposition} 
From the perspective of power transfer efficiency, it is optimal that the harvesters operate at the vertices of the linear regime. This is evident from Proposition \ref{lem:pte M} since ${M}^*_{\rm{PTE}}\in\{M_{{\rm{act}}},M_{{\rm{sat}}}\}$.  
When $M$
is smaller than $M_{{\rm{act}}}$, the power transfer efficiency is (trivially) zero. 
When $M$ is increased beyond $M_{{\rm{sat}}}$, the total power consumption increases while the average harvested power remains the same, reducing the power transfer efficiency. 


\begin{remark}
The PTE-optimal number of antennas depends on the number of users as well as the BS power consumption. 
We note that ${M}^*_{\rm{PTE}}$ increases linearly with the number of users. This is because, with other parameters fixed, adding more users reduces the per user average received energy. Therefore, a larger number of BS antennas are required to activate or saturate the harvesters. 
When $K=1$, it is PTE-optimal to operate with $M_{{\rm{sat}}}$ antennas, since it maximizes the average harvested energy at the user. 
In a multi-user system, however, the condition $
K\geq1+\frac{P_{\rm{TX}}
+P_{\rm{FIX}}}
{P_{\rm{BS}}+{{\tilde{P}_{\rm{CE}}}}+{\tilde{P}_{\rm{LP}}}}$ informs the optimal solution. Here, $\frac{P_{\rm{TX}}
+P_{\rm{FIX}}}
{P_{\rm{BS}}+{{\tilde{P}_{\rm{CE}}}}+{\tilde{P}_{\rm{LP}}}}$ is the ratio of the fixed power consumption to the scalable (circuit/computational) power consumption at the BS -- normalized by the number of antennas.
Note that this condition can be equivalently expressed as a cubic inequality in $K$ as 
\begin{align}
\frac{2B}{S\kappa_{\rm{BS}}}K^3 +\frac{B}{S\kappa_{\rm{BS}}}K^2 + \left(P_{\rm{BS}}\frac{3B}{S\kappa_{\rm{BS}}}\right)K\nonumber\\ - \left(P_{\rm{BS}}+P_{\rm{TX}}
+P_{\rm{FIX}}\right)\geq 0.
\end{align} 
Under realistic assumptions, we obtain some useful insights from this relation. 
  Typically, 
$\frac{P_{\rm{TX}}
+P_{\rm{FIX}}}
{P_{\rm{BS}}+{{\tilde{P}_{\rm{CE}}}}+{\tilde{P}_{\rm{LP}}}}\approx
\frac{P_{\rm{TX}}
+P_{\rm{FIX}}}
{P_{\rm{BS}}}$ since the computational power consumption is usually much smaller than the antenna power consumption.  
This means that when the scalable power consumption exceeds the fixed power consumption, i.e., ${P_{\rm{BS}}}>{P_{\rm{TX}}
+P_{\rm{FIX}}} 
$, it is PTE-optimal to operate with the fewest possible ($M_{{\rm{act}}}$) antennas. This is because the improvement in the harvested energy due to any additional antennas will be overshadowed by the increase in the BS power consumption. Conversely, when the fixed power consumption dominates the scalable power consumption \emph{and} $K< 1+
\frac{P_{\rm{TX}}
+P_{\rm{FIX}}}
{P_{\rm{BS}}}$, it is optimal to operate with $M_{{\rm{sat}}}$ antennas as it maximizes the net harvested energy. 
\end{remark}

\subsubsection*{Optimal K}
We now characterize the number of users ${K}^*_{\rm{PTE}}$ required to optimize the power transfer efficiency, assuming the other parameters to be fixed. In a multi-user system, there is a certain number of users a BS can simultaneously support. If it exceeds this limit, none of the harvesters will be activated due to insufficient received power. We call this limit ${K}_{\rm{max}}$ and express it in terms of the system parameters as
${K}_{\rm{max}}=\left\lfloor\frac{M-1}{\frac{\theta_{\rm{act}}}{P_{\rm{dl}}\beta} -1}\right\rfloor$.
Similarly, we define ${K}_{\rm{sat}}=\left\lfloor\frac{M-1}{\frac{\theta_{\rm{sat}}}{P_{\rm{dl}}\beta} -1}\right\rfloor$ as the maximum number of allowed users such that each harvester operates in the saturated mode. The maximum sum harvested power $B{K}\bar{\delta}$ in the saturated mode is given by $K_{\rm{sat}}\eta_{\rm{EH}}\theta_{\rm{sat}}$. Depending on the system parameters, serving more than $K_{\rm{sat}}$ users may decrease the total harvested power due to a reduction in the per-user harvested power in the non-saturated mode.
\begin{proposition}\label{lem: pte K}
Let us define ${K}^*_{\rm{act}}=\min\left\{{K}_{\rm{max}},\tilde{K}\right\}$
where $\tilde{K}$ is either the integer floor or the integer ceiling of
$\left(M-1\right)\left[1+\sqrt{1+\frac{4}{\left(M-1\right){\dot{P}}_{\rm{CE}}} \left[\frac{P_{\rm{TX}}+P_{\rm{FIX}}+MP_{\rm{BS}}}{M-1}-\dot{P}_{\rm{LP}}\right] }\right]$;
 ${K}^*_{\rm{sat}}=\min\left\{{K}_{\rm{sat}},{\hat{{K}}}\right\}$ where $\hat{K}$ is either the integer floor or the integer ceiling of $\sqrt{\frac{P_{\rm{TX}}+P_{\rm{FIX}}+MP_{\rm{BS}}}{\dot{P}_{\rm{CE}}}}$.
Here, $\dot{P}_{\rm{CE}}=\frac{2MB}{S\kappa_{\rm{BS}}}$ and $\dot{P}_{\rm{LP}}=\frac{3MB}{S\kappa_{\rm{BS}}}$ respectively denote the power required for channel estimation and precoding in a single-user system. 
With an $M$-antenna BS, the PTE-optimal number of users ${K}^*_{\rm{PTE}}$ is given in (\ref{eq:PTE K condition}).
\begin{proof}
See Appendix.
\end{proof}
\end{proposition} 
The condition $K^{*}_{\rm{sat}}>\frac{\left[K_{\rm{act}}^*+M-1\right]\beta P_{\rm{dl}}}{\theta_{\rm{sat}}}
$ implies that the aggregate user harvested power in the saturated mode exceeds that in the non-saturated mode. The expression $\frac{\left(K^{*}_{\rm{sat}}\right)^2\dot{P}_{\rm{CE}}+ K^{*}_{\rm{sat}}\dot{P}_{\rm{LP}}+P_{\rm{TX}}+P_{\rm{FIX}}+MP_{\rm{BS}}}{\left(K^{*}_{\rm{act}}\right)^2\dot{P}_{\rm{CE}}+ K^{*}_{\rm{act}}\dot{P}_{\rm{LP}}+P_{\rm{TX}}+P_{\rm{FIX}}+MP_{\rm{BS}}}$ is the ratio of the BS power consumption in the two modes. When the condition in (\ref{eq:PTE K condition}) holds, it is PTE-optimal to operate the system in the saturated mode since serving more (than $K^{*}_{\rm{sat}}$) users not only increases the BS power consumption, but also reduces the sum harvested power if the harvesters operate in the non-saturated mode.  
Under realistic values of the power consumption model, however, PTE typically improves as more users are added to the system. 
This is because the sum power delivered to the users grows with $K$, despite a decrease in the per-user harvested power. In contrast, the BS power consumption registers only a minor increase since the computational power is negligible compared to the hardware/transmit power. Therefore, it is PTE-optimal to serve the maximum allowed ${K}^*_{\rm{PTE}}={K}_{\rm{max}}$ users in typical systems. 

\subsubsection{Imperfect CSI}
The PTE formulation for the case of imperfect CSI is similar to that of perfect CSI. It is, however, analytically challenging to derive the PTE-optimal solution for this case. Because the two solutions will be qualitatively similar, we do not elaborate this case further. We use simulation results in the following subsection to corroborate this observation.  
\subsection{Simulation Results}\label{secSim:wet} 
We now present the simulation results based on the analysis conducted in this section.
We set the carrier frequency $f_c=1.8$ GHz, coherence time $T_{\rm{c}}=180$ ms, coherence bandwidth $B_{\rm{c}}=10$ kHz, frame size $S=1800$ symbols, system bandwidth $B=1$ MHz, BS transmit power $P_{\rm{dl}}=10$ W, noise power spectral density $\sigma^2=-174$ dBm/Hz, BS computational efficiency $\kappa_{\rm{BS}}=20\times10^9$ flops/W \cite{yael}, BS PA efficiency $\eta_{\rm{PA}}^{\rm{BS}}=0.39$ \cite{ref3}, 
user PA efficiency $\eta_{\rm{PA}}^{\rm{EH}}= 0.3$ \cite{ref3}, BS RF chain power consumption
$P_{\rm{BS}}=1$ W, BS fixed power consumption $P_{\rm{FIX}}=1$ W,
EH conversion efficiency $\eta_{\rm{EH}}= 0.5$, 
EH activation threshold $\theta_{{\rm{act}}}=10~\mu$W, EH saturation threshold $\theta_{{\rm{sat}}}=1~m$W, energy splitting parameter $\xi=0.1$, path-loss exponent $\alpha=3.2$, path-loss intercept $C=1.76\times10^{-4}$ (for a reference distance of 1 m), $r_{\min}=5$ m, and $r_{\max}=20$ m, unless noted otherwise. Geometrically, this setup is equivalent to the case where the users are located on a circle of radius 10.6 m centered at the BS. We set $\alpha_{\rm{Tr}}=\frac{K}{S}$, $\alpha_{\rm{WET}}=1-\alpha_{\rm{Tr}}$, and $\alpha_{\rm{WIT}}=0$.
\subsubsection*{Average harvested power vs. M}
In Fig. \ref{fig: power vs M}, we examine how the average harvested power $B\bar{\delta}_i$ at a user varies as a function of the number of BS antennas $M$ and users $K$. We consider the ideal case where the BS has perfect channel knowledge, and the realistic case where it has imperfect channel knowledge due to LS/MMSE estimation. 
We obtain the analytical (anl) results using Lemma $\ref{lem:harperfect}-\ref{lem:ls harv}$, and the simulation (sim) results using Monte Carlo simulations for $10^4$ trials. 
We include the results for both single-user ($K=1$) and multi-user systems ($K=2$). 

We can draw several useful insights from Fig. \ref{fig: power vs M}. First, with $K$ fixed and $M\geq M_{\rm{act}}$ (which is $7$ for $K=1$ and $13$ for $K=2$), the average harvested power increases\footnote{As evident from Corollary \ref{cor:m/k ratio}, this trend does not hold when $M$ and $K$ are scaled proportionally (for a fixed $M/K$) as the harvested power remains almost constant.} with $M$, until the harvester saturates. 
For example, for $K=1$, the harvested power saturates at $M=M_{\rm{sat}}=1089$ antennas (and at $M=2177$ antennas for $K=2$ which is not shown). 
Second, adding more users reduces the harvested energy at each user. This is due to a reduction in the energy beamformed at each user. Third, imperfect channel knowledge causes only a minor loss in the harvested energy. Similar trend was reported in prior work \cite{kayshap2015massivewet}. This suggests that the insights drawn with perfect channel knowledge may be applicable to realistic scenarios with imperfect channel knowledge. 


\begin{figure} [t]
	\centerline{
		\includegraphics[width=\columnwidth]{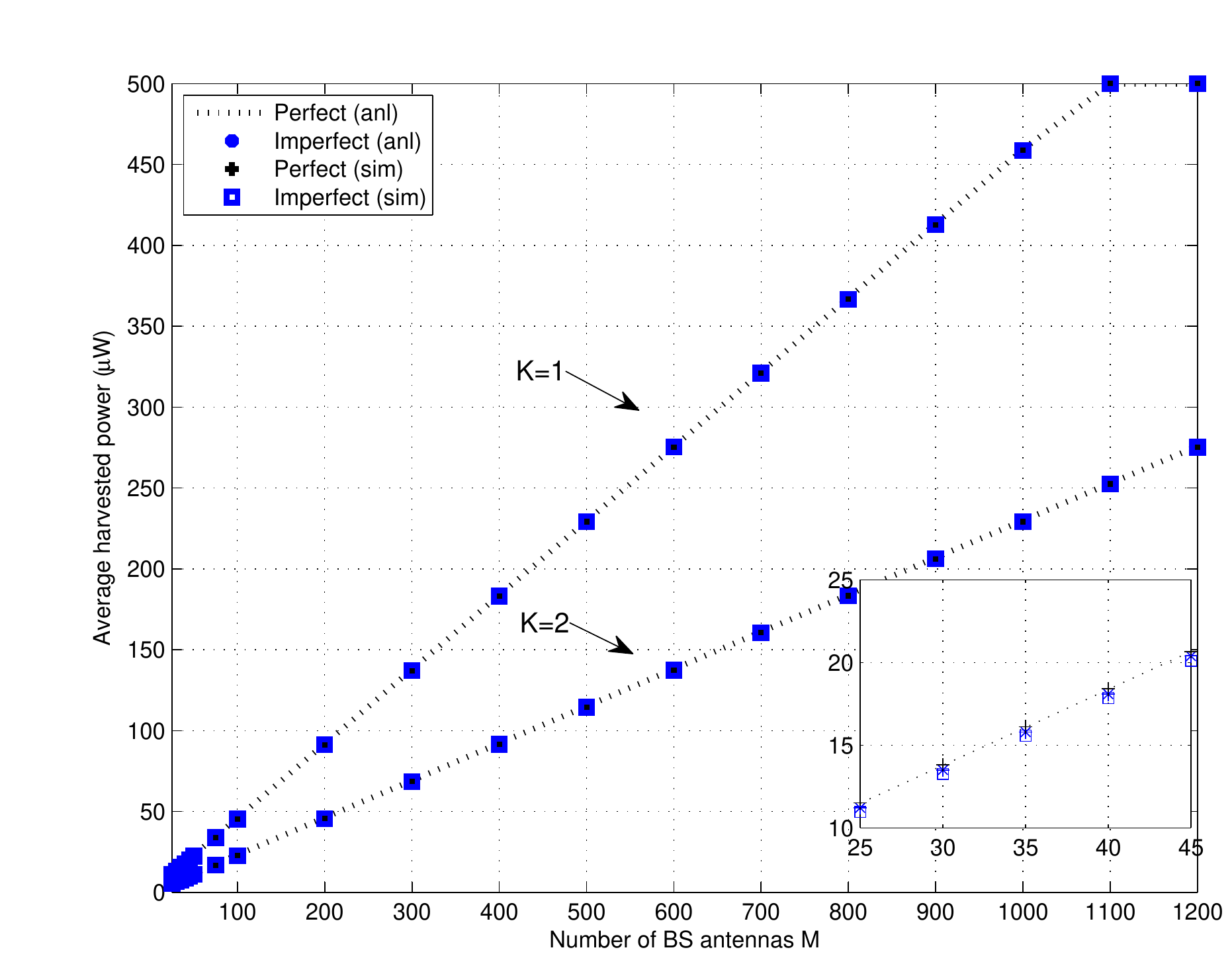}
	}
	\caption{The average harvested power $B\bar{\delta}_i$ increases with the number of BS antennas $M$, and decreases with the number of users $K$. The inset shows a zoomed-in version of the curves for $K=1$. Imperfect channel knowledge (LS/MMSE channel estimation) causes a minor degradation versus perfect channel knowledge. Simulation-based (sim) results validate the analytical (anl) results.}
	\label{fig: power vs M}
\end{figure}

\subsubsection*{PTE vs. M}
In Fig. \ref{fig: pte vs M}, we plot the power transfer efficiency versus the number of BS antennas for single-user and multi-user systems. It reveals how PTE behaves in terms of key system parameters, confirming the insights drawn in Proposition \ref{lem:pte M}. First, we observe that there is an optimal $M$ that maximizes the PTE.  
In a single-user system, it is optimal to operate with the maximum possible antennas $M_{\rm{sat}}$ in the linear regime, beyond which the PTE tends to decrease. In a multi-user system, this is not necessarily the case. For example,
 in the considered multi-user case $K=40$, the PTE is maximized using the fewest possible $M_{\rm{act}}$ antennas. This is because, for this example, the boost in harvested power due to additional antennas is overshadowed by the increase in the BS circuit power consumption. With fewer than $M_{\rm{act}}$ antennas, the PTE is zero as the received power fails to meet the activation threshold. Second, we observe that imperfect channel knowledge results in a minor degradation in PTE. 
 Moreover, imperfect knowledge requires a larger number of antennas to activate the system compared to the case with perfect channel knowledge. 
 Third, we observe that multi-user system yields a higher PTE than the single-user system. This is because the sum harvested power increases as more users are added to the system, despite a decrease in the per user harvested power.
 This trend holds as long as the number of users do not exceed $K_{\max}$, beyond which the individual harvested power -- and therefore the PTE -- drops to zero. This trend is in line with Proposition \ref{lem: pte K}.

\begin{figure} [t]
	\centerline{
\includegraphics[width=\columnwidth]{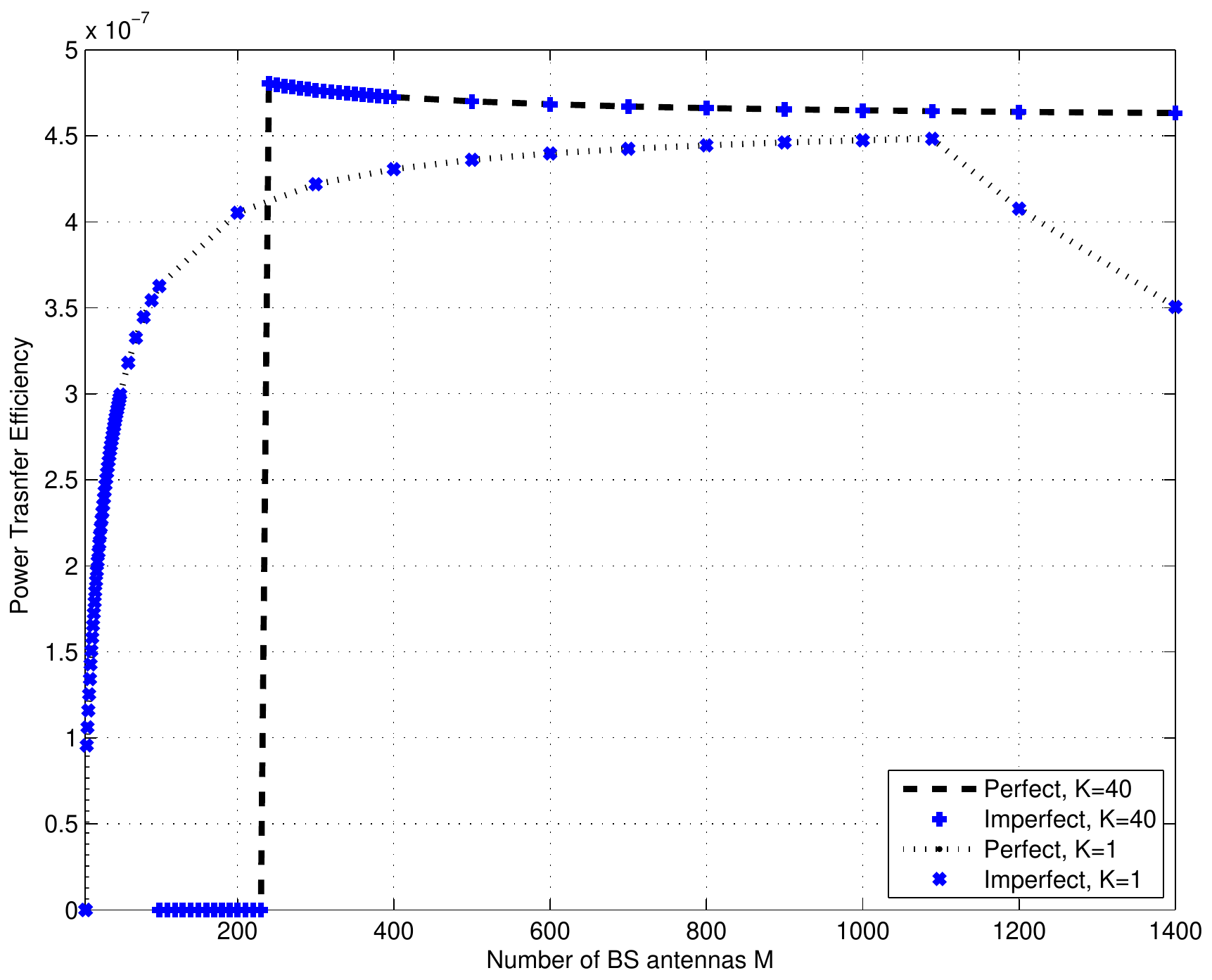}
	}
	\caption{Power transfer efficiency $\rm{PTE}$ vs. the number of BS antennas $M$ for $K=1$ and $K=40$ users. There is an optimal $M$ that maximizes the PTE, as reported in Proposition \ref{lem:pte M}: 
	For single-user system, PTE is optimized by operating with maximum possible antennas in the linear regime.
	For the considered multi-user system, operating with fewest possible antennas maximizes the PTE.
 Imperfect channel knowledge (LS/MMSE channel estimation) causes a negligible degradation versus perfect channel knowledge.}
	\label{fig: pte vs M}
\end{figure}

\subsubsection*{Optimizing energy splitting parameter $\xi$}
We now consider the effective harvested power at a user, after discounting the amount used for uplink pilot transmission (see Remark \ref{rem: eff}).
In 	Fig. \ref{fig: opt energy alloc}, we set $M=500$, $S=100$, $P_{\rm{dl}}=20$ W, $r_{\min}=5$ m, $r_{\max}=50$ m, and plot the effective harvested power at a user versus the energy splitting parameter $\xi$ for various values of $K$. 
We recall that the energy splitting parameter $\xi$ gives the fraction of the harvested energy that a user reserves for pilot transmission. 
For the considered system, we observe that dedicating around $1\%$ of the harvested power maximizes the effective power available to the user. We note that this fraction will be even smaller in systems with a larger frame size.
Moreover, deviating from this optimal value may cause a significant degradation in effective harvested power: 
allocating a smaller fraction reduces the uplink transmit power, which decreases the channel estimation accuracy at the BS. The resulting BS transmission based on inaccurate channel knowledge sacrifices the beamforming gain, which reduces the harvested power at the user. Conversely, allocating more energy for uplink training improves the harvested power. This improvement, however, is insufficient to justify the underlying increase in the uplink transmit power. As a result, the effective harvested power will reduce nonetheless. This trend is in line with the discussion in Remark \ref{rem: eff}. 
We further note that the optimal value is insensitive to the number of users. This suggests that a user does not need to tune this parameter when other users enter or leave the system. Finally, we observe that the effective harvested power at a user decreases as more users are being served. As explained earlier, this is due to a reduction in the energy beamformed to each user. 
\begin{figure} [t]
	\centerline{
		\includegraphics[width=\columnwidth]{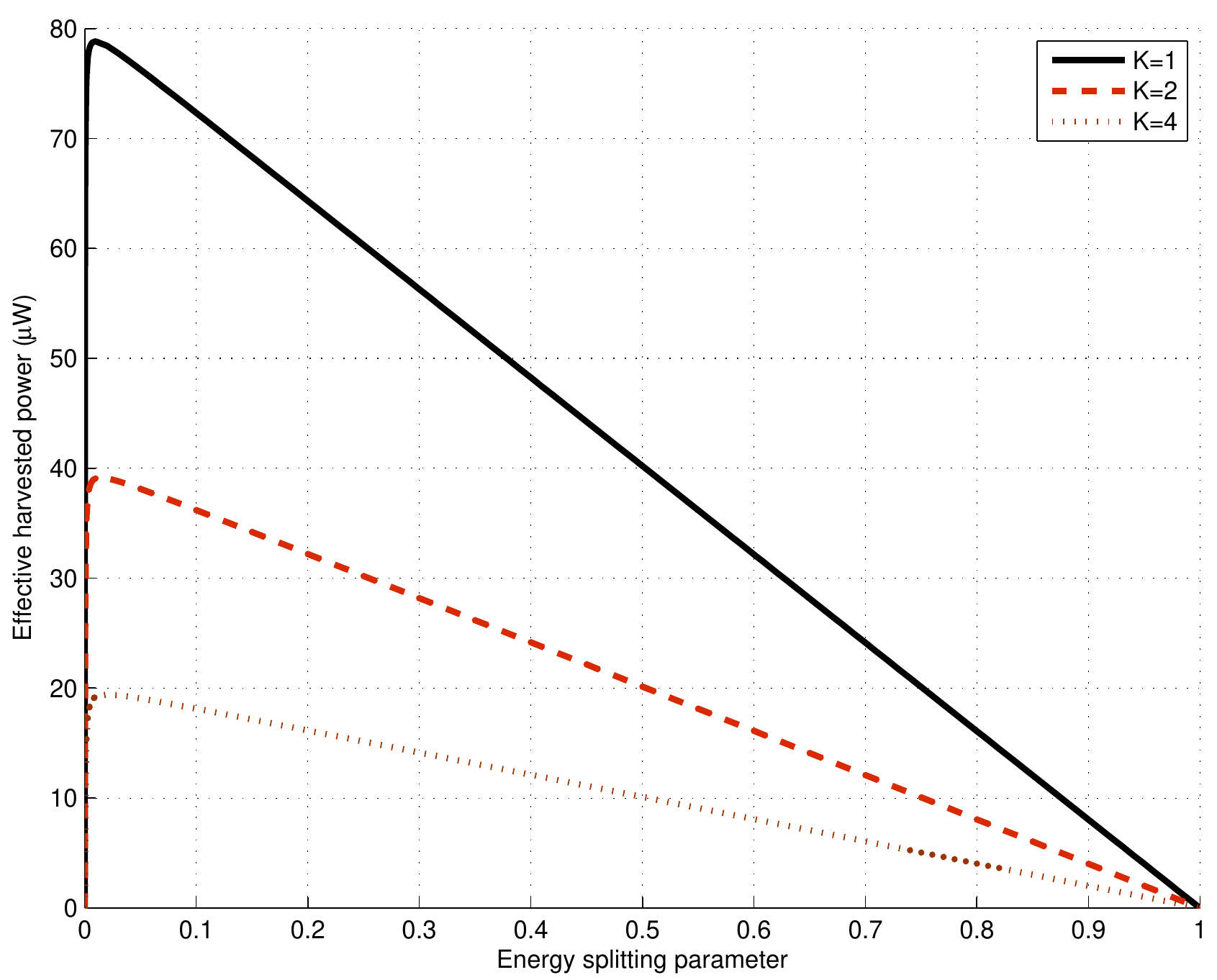}
	}
	\caption{Effective harvested power vs. energy splitting parameter $\xi$ for $M=500$. A user can maximize the effective harvested power by allocating the right amount of harvested energy for uplink pilot transmission. The optimal value is not particularly sensitive to the number of users in the system. The harvested power decreases as more users are added to the system.}
	\label{fig: opt energy alloc}
\end{figure}

\section{Wireless Energy and Information Transfer}\label{sec WIPT}
In this section, we consider wireless-powered communications where the BS charges users in the downlink, and the users leverage the harvested energy to communicate with the BS on the uplink. This is different from the previous section where no information transfer was considered on the uplink.
\begin{figure*}[!t]
	\normalsize
	\setcounter{mytempeqncnt}{\value{equation}}
	\setcounter{equation}{22}

\begin{align}\label{eq:uplink achievable rate}
{R}_i=
\begin{cases} 
0 & M<M_{{\rm{act}},i}\\
\alpha_{\rm{WIT}}B\log_2\left(1+\left(1-\xi_i\right){\beta_i}\,\eta^{\rm{EH}}_{\rm{PA}}\,\frac{{\eta_{\rm{EH}}}{\bar{\gamma}}_i}{\alpha_{\rm{WIT}}\sigma^2} \left(M-K\right)\right), & \hspace{-0in}M_{{\rm{act}},i}\leq M< M_{{\rm{sat}},i}\\
\alpha_{\rm{WIT}}B\log_2\left(1+\left(1-\xi_i\right)\beta_i\,\eta^{\rm{EH}}_{\rm{PA}}\frac{\eta_{\rm{EH}}\theta_{{\rm{sat}}}}{B\alpha_{\rm{WIT}}\sigma^2} \left(M-K\right)\right), & M\geq M_{{\rm{sat}},i}
\end{cases}
\end{align}
	\hrulefill
	\vspace*{4pt}
\end{figure*}
\subsection{Uplink Achievable Rate}
We now provide analytical expressions for the uplink achievable rate for a wirelessly powered user. Note that the total harvested energy is used for sending both training and data symbols during uplink transmission.  
As defined previously, $\xi_i\in (0,1)$ is the fraction of the harvested energy $\bar{\delta}_i$ at a user $i$ reserved for the uplink pilot transmission, while the remaining fraction $1-\xi_i$ of the harvested energy is used for uplink data transmission.
We assume that a user $i$ transmits uplink data symbols with an average energy $p^i_{\rm{ul}}=\frac{\eta^{\rm{EH}}_{\rm{PA}}\left(1-\xi_i\right)\bar{\delta}_i}{\alpha_{\rm{WIT}}}$ (in joules/symbol). 
Here, $\eta^{\rm{EH}}_{\rm{PA}}\in(0,1]$ denotes the user PA efficiency. At the users, while we explicitly model only the transmit power consumption, any additional power consumption (e.g., due to computation) could be equivalently handled by tuning (i.e., further reducing) the parameter $\eta^{\rm{EH}}_{\rm{PA}}$. 
While we account for the training overhead, we ignore the loss in harvested power due to channel estimation errors. This simplifies the analysis, and could be justified since the imperfect channel knowledge causes only a minor loss in the harvested power (Section \ref{secSim:wet}). 
For uplink detection, we assume that the BS uses a Zero-forcing (ZF) receive filter. 
Leveraging the convexity of the function $\log(1+x^{-1})$ where $x>0$, we use Jensen's inequality to obtain a lower bound on the ergodic uplink achievable rate, and call it the achievable rate in the ensuing analysis\cite{Ngo2013EE}. 
\begin{lemma}
With an $M$-antenna BS serving $K<M$ users, the uplink achievable rate ${R}_i$ for a remotely-powered user $i$ is given by (\ref{eq:uplink achievable rate}),
where $\bar{\gamma}_i$ is the average received energy as defined in Lemma \ref{lemma:avgrcvperfectcsi}. When $\zeta_i=\frac{1}{K}$, the achievable rate for $M_{{\rm{act}},i}\leq M< M_{{\rm{sat}},i}$ can be further simplified to
\begin{align}\label{eq:rate}
R_i=\alpha_{\rm{WIT}}B\log_2\left(1+\rho_i \left[M-K\right]\left[1+\frac{M-1}{K}\right]\right),
\end{align}
where 
\begin{align}\label{eq:rhoi}
\rho_i&\triangleq\frac{\left(1-\xi_i\right)p_{\rm{dl}}\,\alpha_{\rm{WET}}\,\eta_{\rm{EH}}\,\eta^{\rm{EH}}_{\rm{PA}}\,\beta_i^2}{\alpha_{\rm{WIT}}\sigma^2}\nonumber\\
&=
\frac{P_{\rm{dl}}\left(1-\xi_i\right)\,\,\eta_{\rm{EH}}\,\eta^{\rm{EH}}_{\rm{PA}}\,\beta_i^2}{B\alpha_{\rm{WIT}}\sigma^2}.
\end{align}
captures the effect of the system parameters (other than $M$ and $K$) on the uplink SNR.
\begin{proof}
This follows by noting that $\frac{p^i_{\rm{ul}}}{{\sigma}^2}=\frac{\eta^{\rm{EH}}_{\rm{PA}}\left(1-\xi_i\right)\bar{{\delta}}_i}{\alpha_{\rm{WIT}}\sigma^2}$ is the uplink (transmit) signal-to-noise ratio (SNR) for user $i$, applying Lemma 1, and invoking the result in \cite[Proposition 3]{Ngo2013EE}.
\end{proof}
\end{lemma}
We note that the uplink rate is expressed in terms of the downlink transmit symbol energy $p_{\rm{dl}}$ since the user exploits the harvested energy to power its uplink transmission. This further means that the uplink communication link will typically operate in the low-SNR regime due to the limited energy available at the user. This case is further elaborated in the following corollary. 
\begin{corollary}
Let us consider a user in the linear mode such that $M_{{\rm{act}},i}\leq M< M_{{\rm{sat}},i}$. In the low-SNR regime, i.e., when $R_i\approx \alpha_{\rm{WIT}}B\rho_i \left(M-K\right)\left(1+\frac{M-1}{K}\right)$, the uplink achievable rate is the most susceptible to path loss\,\textemdash\,being proportional to the \emph{square} of the large-scale channel gain $\beta_i$. Fortunately, additional antennas are the most beneficial also in this regime as the rate approximately grows with the square of $M$. 
\begin{proof}
This follows from (\ref{eq:rhoi}) and by noting that $\log(1+x)=x+\mathcal{O}\left(x^2\right)$ for $|x|\leq 0.5$.
\end{proof}
\end{corollary}

\begin{remark}\label{rem: rate growth}
The achievable rate reports a faster growth with $M$ in the non-saturated mode than the saturated mode. This follows from (\ref{eq:uplink achievable rate}) by noting that the effective uplink SNR grows approximately with the square of $M$ in the non-saturated mode, but only linearly in the saturated mode. 
In the linear mode, more antennas help improve the downlink energy transfer as well as the uplink information detection.
This is not the case in the saturated mode where only the uplink detection benefits from more antennas. We further note that, unlike the harvested power, the achievable rate does \emph{not} saturate in the saturated mode.  
\end{remark}
\subsection{Energy Efficiency}
We now characterize the total energy efficiency of the considered system by leveraging the power consumption model used in (\ref{eq:PTE}).
Similar to Section \ref{sec:PTE}, 
we set $\beta_i~\forall~i\in\mathcal{I}_K$ 
to the average
$\mathbb{E}[\beta_i]=C\mathbb{E}[d_i^{-\alpha}]\triangleq\beta$, where $\mathbb{E}[d_i^{-\alpha}]$ is given in Section \ref{Sec: sytem model}. 
We assume all users have the same value for the energy splitting parameter $\xi_i\triangleq\xi$ such that $\rho_i\triangleq\rho$. With these simplifying assumptions, the users achieve an identical average rate, i.e., $R_i\triangleq R~\forall~i\in\mathcal{I}_K$.
We define the total {energy efficiency} $\left(\text{EE}\right)$ of the overall system (in bits/joule) as the ratio of the average uplink sum rate to the total average power consumed, i.e., 
\allowdisplaybreaks{
\begin{align}\label{eq:EE}
\text{EE}&\left(M,K\right)=\frac{K R}{P_{\rm{TX}}+P_{\rm{c}}}\nonumber\\
&=\frac{K R}{
P_{\rm{TX}}
+P_{\rm{FIX}}+ MP_{\rm{BS}}+P_{\rm{CE}}+P_{\rm{LP}}+P_{\rm{DEC}}KR
}
\end{align}}where 
the power consumption model is similar to (\ref{eq:PTE}) except for two components: i) $P_{\rm{LP}}$ is modified to account for additional BS linear processing, i.e., in addition to the power required for computing the downlink precoder $\left(\frac{3MKB}{S\kappa_{\rm{BS}}}\right)$, it also includes the power required for computing the uplink ZF filter $\left(\frac{B(\frac{K^3}{3}+3MK^2+MK)}{S\kappa_{\rm{BS}}}\right)$ once per coherence block, and for evaluating a matrix-vector multiplication for each data symbol $\frac{2\alpha_{\rm{WIT}}MKB}{\kappa_{\rm{BS}}}$\cite{ref3,boyd2004convex}; ii) $P_{\rm{DEC}}KR$ is introduced to model the power consumed in decoding the received data, where $P_{\rm{DEC}}$ parameterizes the BS decoder power consumption (in W/bit/s)\cite{ref3}. 
These terms were absent in (\ref{eq:PTE}) since uplink data transmission was not considered. 
We note that the computational power consumption is usually negligible compared to the antenna power consumption in the large-antenna regime.
Moreover, the uplink transmit power, which is a fraction of the average harvested power, only appears in the numerator (via the expression for the achievable rate $R$). This is because the energy harvesting users do not have any power source except for the wireless energy delivered by the BS.
\begin{remark}\label{rem: EE vanishes}
We observe from (\ref{eq:EE}) that EE eventually vanishes in the large $M$ regime. This is because the data rate in the numerator grows only logarithmically whereas the power consumption in the denominator grows linearly with $M$.
\end{remark}

 \subsubsection*{Energy Efficiency Optimization}
We now characterize the optimal BS transmit power that maximizes the total energy efficiency for a given number of antennas and users. 
\begin{lemma}\label{lem: optimal power}
	When the harvesters operate in the non-saturated mode, the EE-optimal transmit power at the BS is given by $P_{\rm{dl}}^*={{\alpha}_{\rm{WET}}Bp^*_{\rm{dl}}}$ where
	\begin{align}\label{eq:pdl}
	p^*_{\rm{dl}}=&\frac{e^{1+\text{W}\left[\frac{{\eta^{\rm{BS}}_{\rm{PA}}}\tilde{\rho}\left(\tilde{C}+M\tilde{D}\right)\left(M-K\right)\left(1+\frac{M-1}{K}\right)}{{eB\alpha_{\rm{WET}}}}-\frac{1}{e}\right]}-1}{\tilde{\rho}\left(M-K\right)\left(1+\frac{M-1}{K}\right)},
	\end{align}
	and $\text{W}[\cdot]$ is the Lambert-W function. The constants
	$\tilde{\rho}=\frac{\left(1-\xi_i\right)\,\alpha_{\rm{WET}}\,\eta_{\rm{EH}}\,\eta^{\rm{EH}}_{\rm{PA}}\,\beta_i^2}{\alpha_{\rm{WIT}}\sigma^2}$,
	 $\tilde{C}=P_{\rm{FIX}}+\frac{BK^3}{3S\delta_{\rm{BS}}}$, and $\tilde{D}=P_{\rm{BS}}+\frac{2B}{\delta_{\rm{BS}}}(1+\frac{2}{S})K+\frac{3B}{S\delta_{\rm{BS}}}K^2$. 
\end{lemma}
\begin{proof}
See Appendix.
\end{proof}

\begin{remark}
 The expression in (\ref{eq:pdl}) is applicable when the harvesters operate in the linear mode. 
 When $M<M_{\rm{act}}$, we should increase the BS transmit power to activate the harvesters, resulting in a non-zero EE,
 i.e., set $P_{\rm{dl}}=
  \frac{\theta_{{\rm{act}}}K}{\beta\left(M+K-1\right)}$. 
  Similarly, when $P^*_{\rm{dl}}$ satisfies $M>M_{{\rm{sat}}}$, 
  reducing the transmit power to be at least as small as  $\frac{\theta_{{\rm{act}}}K}{\beta\left(M+K-1\right)}$ helps improve the energy efficiency. This is because once the harvesters get saturated, the excess power only increases the power consumption without bringing any improvement in the achievable rate. Using this principle, Algorithm \ref{talha} provides a heuristic procedure for selecting the transmit power for an energy efficient operation. We validate this approach by exhaustively searching for the optimal transmit power in the next section.   
  \end{remark}
  
  \begin{remark}
  We note that it is EE-optimal to dedicate most of the resources for the uplink, i.e., $\alpha_{\rm{WIT}}$ should be much higher than $\alpha_{\rm{WET}}$. This is because the uplink sum rate tends to improve as $\alpha_{\rm{WIT}}$ is increased. This can be observed by plugging the EE-optimal transmit power from Lemma \ref{lem: optimal power} in the achievable rate expression in (\ref{eq:uplink achievable rate}). The prelog term of the resulting rate expression grows linearly, whereas the logarithmic term decays only sublinearly with $\alpha_{\rm{WIT}}$. Furthermore, the EE-optimal transmit power also increases to compensate for a reduction in $\alpha_{\rm{WET}}$. The increase in the sumrate overpowers the increase in total power consumption, leading to an increase in the overall energy efficiency. 
  \end{remark}

\begin{algorithm}
	\caption{Selecting nearly optimal $P_{\rm{dl}}$}\label{talha}
	\begin{algorithmic}[1]
		\Procedure{}{}
		\State $p_{{\rm{act}}}\gets\frac{K\theta_{{\rm{act}}}}{\alpha_{\rm{WET}}B\beta(M+K-1)}$,
		$p_{{\rm{sat}}}\gets\frac{K\theta_{{\rm{sat}}}}{\alpha_{\rm{WET}}B\beta(M+K+1)}$			
		\State $p_{\rm{dl}} \gets  
		\frac{e^{1+\text{W}\left[\frac{{\eta^{\rm{BS}}_{\rm{PA}}}\tilde{\rho}\left(\tilde{C}+M\tilde{D}\right)\left(M-K\right)\left(1+\frac{M-1}{K}\right)}{{eB\alpha_{\rm{WET}}}}-\frac{1}{e}\right]}-1}{\tilde{\rho}\left(M-K\right)\left(1+\frac{M-1}{K}\right)}
		$		
		
		\State $M_{{\rm{act}}}\gets\lceil\frac{K\theta_{{\rm{act}}}}{\alpha_{\rm{WET}}B\beta p_{\rm{dl}}}-(K-1)\rceil$	
		\State 
		$M_{{\rm{sat}}}\gets\lfloor\frac{K\theta_{{\rm{sat}}}}{\alpha_{\rm{WET}}B\beta p_{\rm{dl}}}-(K-1)\rfloor$	
		
		\State \textbf{If} $M >M_{{\rm{sat}}}$ 
		\State \quad$p_{\rm{dl}}\gets \min\left(p_{{\rm{sat}}},p_{\rm{dl}}\right)$
		\State \textbf{else} 
		\State \quad$p_{\rm{dl}}\gets \max\left(p_{{\rm{act}}},p_{\rm{dl}}\right)$
		\State \textbf{end}
		\State $P_{\rm{dl}}^*\gets \alpha_{\rm{WET}}p_{\rm{dl}}B$
		
		\EndProcedure
	\end{algorithmic}
\end{algorithm}

\subsection{Simulation Results}\label{SecSim}
We now present simulation results to verify the analytical insights in this section. The simulation parameters are the same as described in Section \ref{secSim:wet}, unless noted otherwise. We set $\alpha_{\rm{WET}}=0.01$ such that $\alpha_{\rm{WIT}}=1-\alpha_{\rm{WET}}-\alpha_{\rm{Tr}}\approx0.98$. We dedicate more resources for the uplink because energy efficiency benefits from increasing $\alpha_{\rm{WIT}}$ thanks to an increase in the uplink sum rate. 
We set $P_{\rm{FIX}}=18$ W (we expect an increased fixed power consumption at the BS compared to the scenario in Section \ref{secSim:wet} since it now has to deal with data reception on the uplink similar to a traditional BS)\cite{ref3}, 
$P_{\rm{DEC}}=10^{-9}$ W/bits/sec \cite{ref3}, $r_{{\min}}=5$ m, and $r_{{\max}}=50$ m. Geometrically, this is equivalent to the setup where users are located on a circle of radius 18.3 m centered at the BS. 
In the following figures, \emph{ideal} curve is for ideal (linear) energy harvesters with $\theta_{{\rm{act}}}\rightarrow 0$ and $\theta_{{\rm{sat}}}\rightarrow \infty$. Similarly, \emph{practical} curve represents the case where practical energy harvesters are deployed in a system optimized for ideal energy harvesters.

\emph{{EE}-optimal BS transmit power vs. M:} In Fig. \ref{fig:pwr}, we plot the EE-optimal transmit power against $M$ for both ideal and practical energy harvesters for $K=2$ users. We note that the EE-optimal transmit power selection assuming ideal energy harvesters can be misleading for practical energy harvesters. When $M$ is large (say $> 2000$ in Fig. \ref{fig:pwr}), it is EE-optimal to reduce the transmit power with $M$ for practical harvesters. This helps avoid energy wastage when the energy harvesters operate in the saturated mode. This is contrary to the ideal case where the EE-optimal transmit power increases with $M$.
When $M$ is small (say $< 90$), it is EE-optimal to use a larger transmit power than the ideal case. A sufficient increase in the transmit power helps activate the nodes, resulting in a nonzero data rate and EE. We also note that the transmit power selection based on Algorithm \ref{talha} closely approximates the optimal solution. Similarly, Fig. \ref{fig:per-antenna} shows how the per-antenna transmit power scales with $M$. Unlike the total transmit power which may increase with $M$, the EE-optimal per-antenna transmit power typically reduces as $M$ is increased.  

\emph{Maximal EE vs. M:} In Fig. \ref{fig:EE}, we plot the maximal EE versus $M$ for $K=2$ users. We observe that a system designed for ideal energy harvesters may suffer a severe performance loss when used with practical energy harvesters. For example, when $M$ is small (say $< 90$ in Fig. \ref{fig:pwr} and \ref{fig:EE}), the EE is zero for practical energy harvesters. This is because the transmit power, though EE-optimal for an ideal harvester, is insufficient to activate a practical harvester.   
With a practical harvester, EE-optimality warrants increasing the transmit power so as to improve the uplink rate and EE. 
Conversely, when $M$ is large (say $> 2000$ in Fig. \ref{fig:pwr} and \ref{fig:EE}), the maximal EE is attained by sufficiently reducing the transmit power to avoid saturating the harvesters. For $K=2$ users, EE-optimality is achieved using around 56 antennas.
In Fig. \ref{fig:EEK50}, we observe similar trends for the case of $K=50$ users where EE-optimality is achieved using around 230 antennas. 
We note that the EE achieved using Algorithm \ref{talha} closely approximates the optimal solution. 
Further, in line with Remark \ref{rem: EE vanishes}, EE will eventually vanish in the large antenna regime due to excessive power consumption. 
Finally, 
we observe that the EE-optimal operating point indeed lies in the massive antenna regime. 

\emph{Uplink sum rate vs. M:} In Fig. \ref{fig:rate}, we plot the uplink sum rate obtained using the EE-optimal policy considered in the previous figures. In contrast to the EE, the sum rate improves monotonically with $M$. Further, the rate grows with $M$ at a slower pace in the saturated mode (Remark \ref{rem: rate growth}). This is because the harvested power does not increase with $M$ in the saturated mode, leaving only the uplink detection to benefit from additional BS antennas.

\begin{figure} [t]
	\centerline{
		\includegraphics[width=\columnwidth]{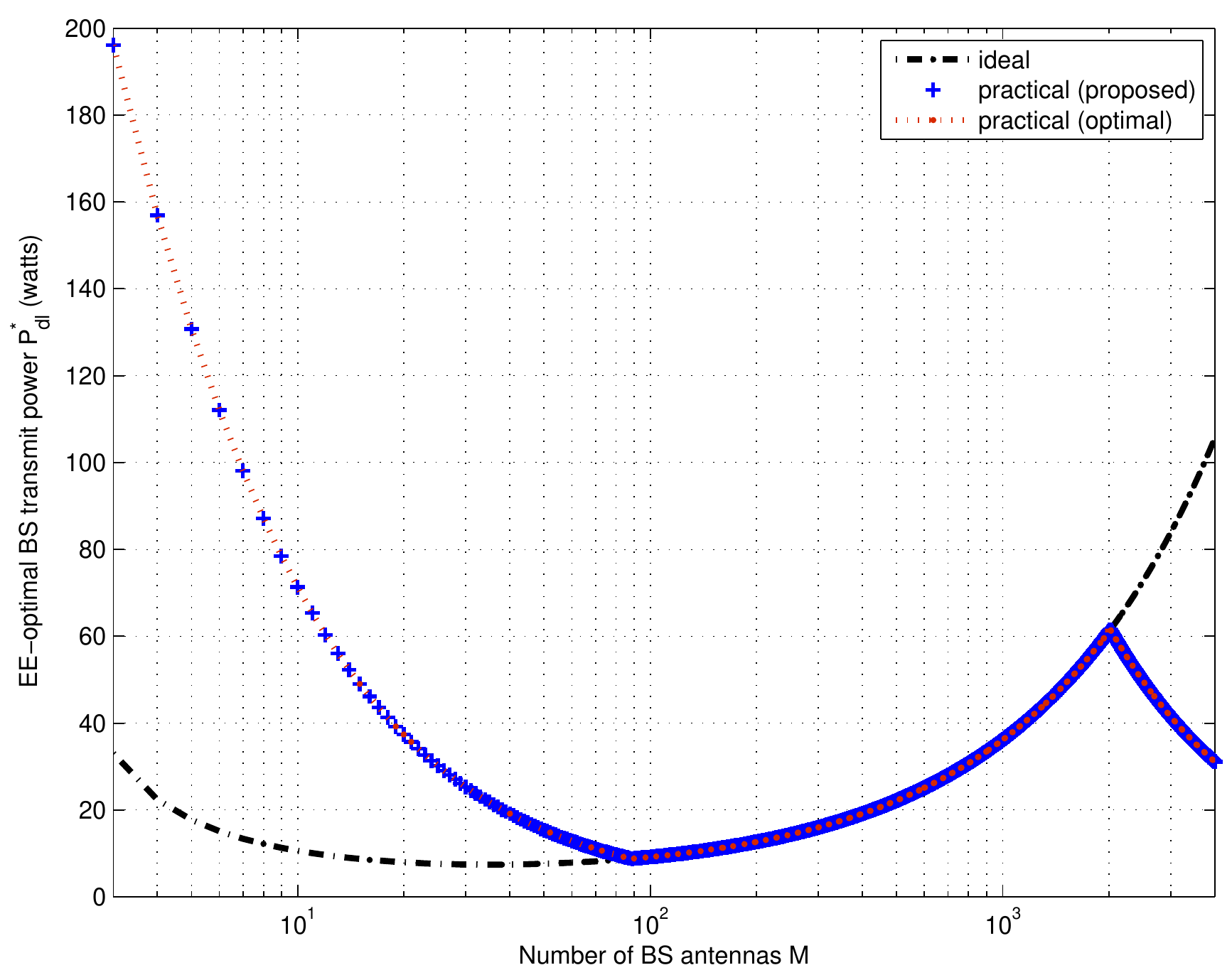}
	}
	\caption{EE-optimal transmit power vs. the number of BS antennas for ideal as well as practical energy harvesters. 
	The EE-optimal approach for the ideal (linear) case could be misleading for the practical (nonlinear) case: 
	It is EE-optimal to i) sufficiently increase the transmit power to wake up the users; and ii) decrease it in the saturated mode to avoid energy wastage. The proposed approach  closely approximates the optimal solution.}
	\label{fig:pwr}
\end{figure}
\begin{figure} [t]
	\centerline{
		\includegraphics[width=\columnwidth]{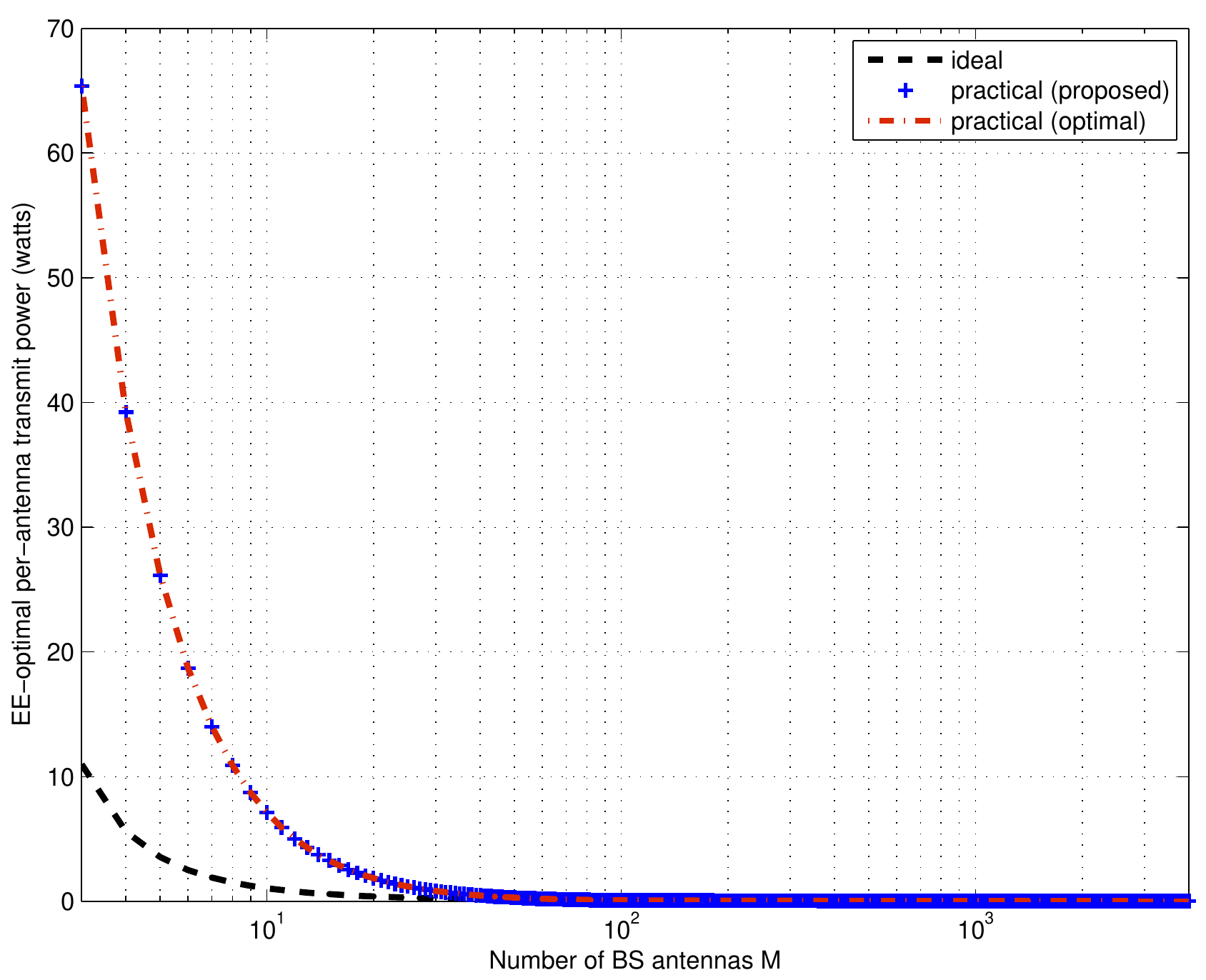}
	}
	\caption{The EE-optimal per-antenna transmit power vs. the number of BS antennas. 
		The per-antenna optimal transmit power tends to decrease with $M$ for both ideal and practical cases.}
	\label{fig:per-antenna}
\end{figure}

\begin{figure} [t]
	\centerline{
		\includegraphics[width=\columnwidth]{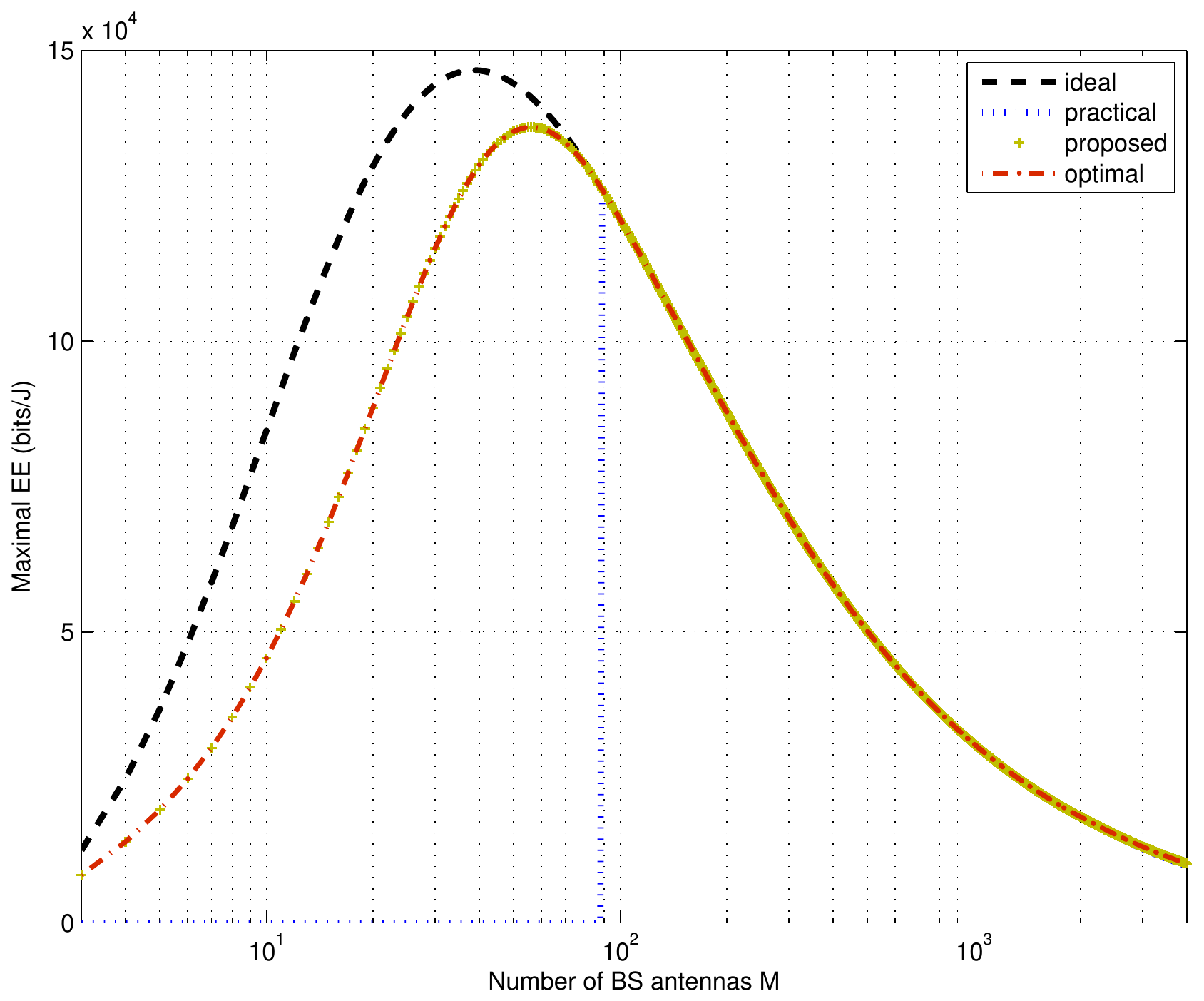}
	}
	\caption{The maximal EE vs. the number of BS antennas for $K=2$ users. A comparison between ``ideal" and ``practical' shows the performance actually achieved with practical energy harvesters in a system designed for ideal (linear) energy harvesters. 
	The EE-optimal approach for the ideal case could be very misleading for the practical case. Note that the proposed solution significantly improves the EE, and closely approximates the optimal solution. Moreover, there exists an optimal $M$ that maximizes the EE.}
	\label{fig:EE}
\end{figure}

\begin{figure} [t]
	\centerline{
		\includegraphics[width=\columnwidth]{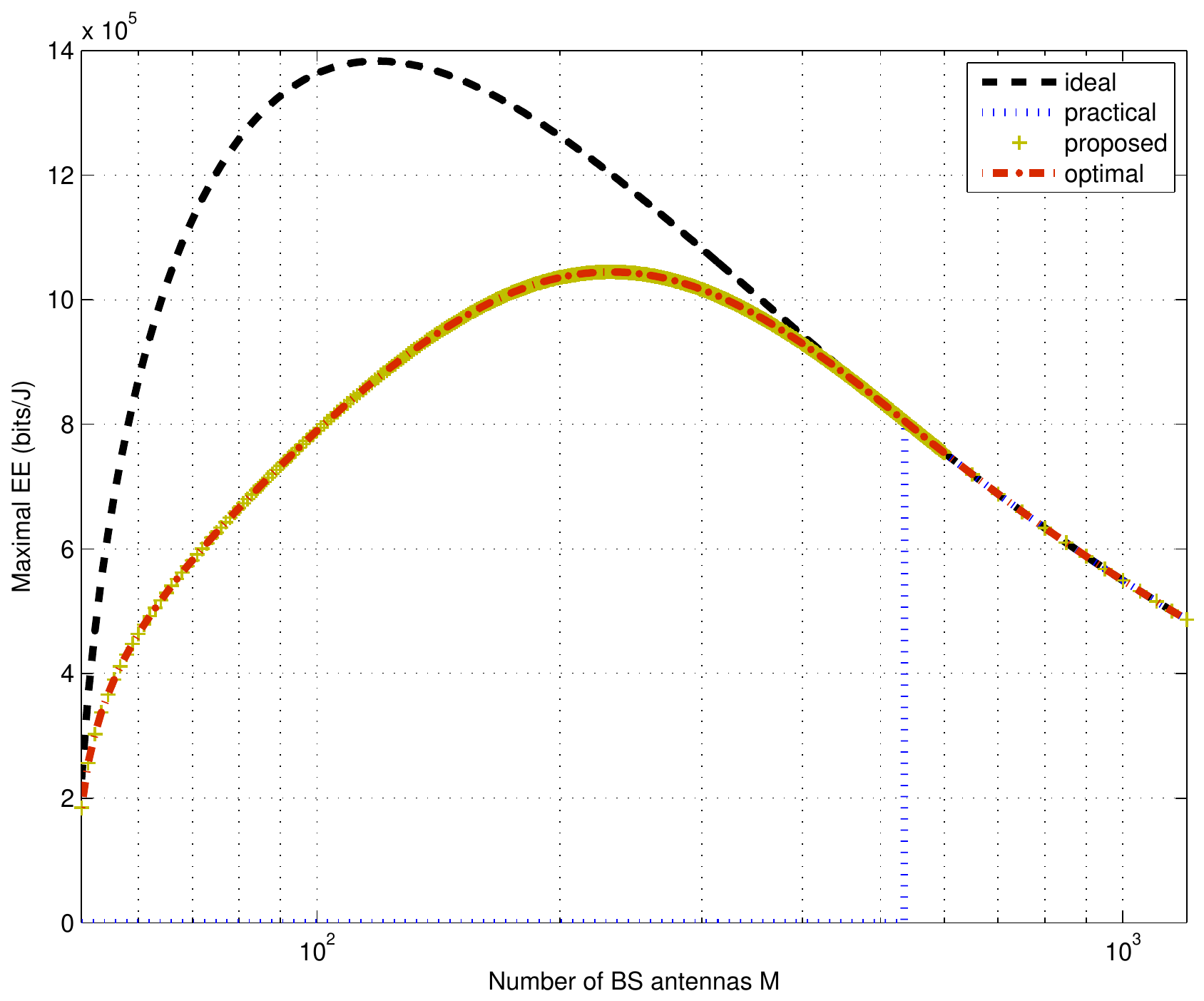}
	}
	\caption{The maximal EE vs. the number of BS antennas for $K=50$ users.}
	\label{fig:EEK50}
\end{figure}

\begin{figure} [t]
	\centerline{
		\includegraphics[width=\columnwidth]{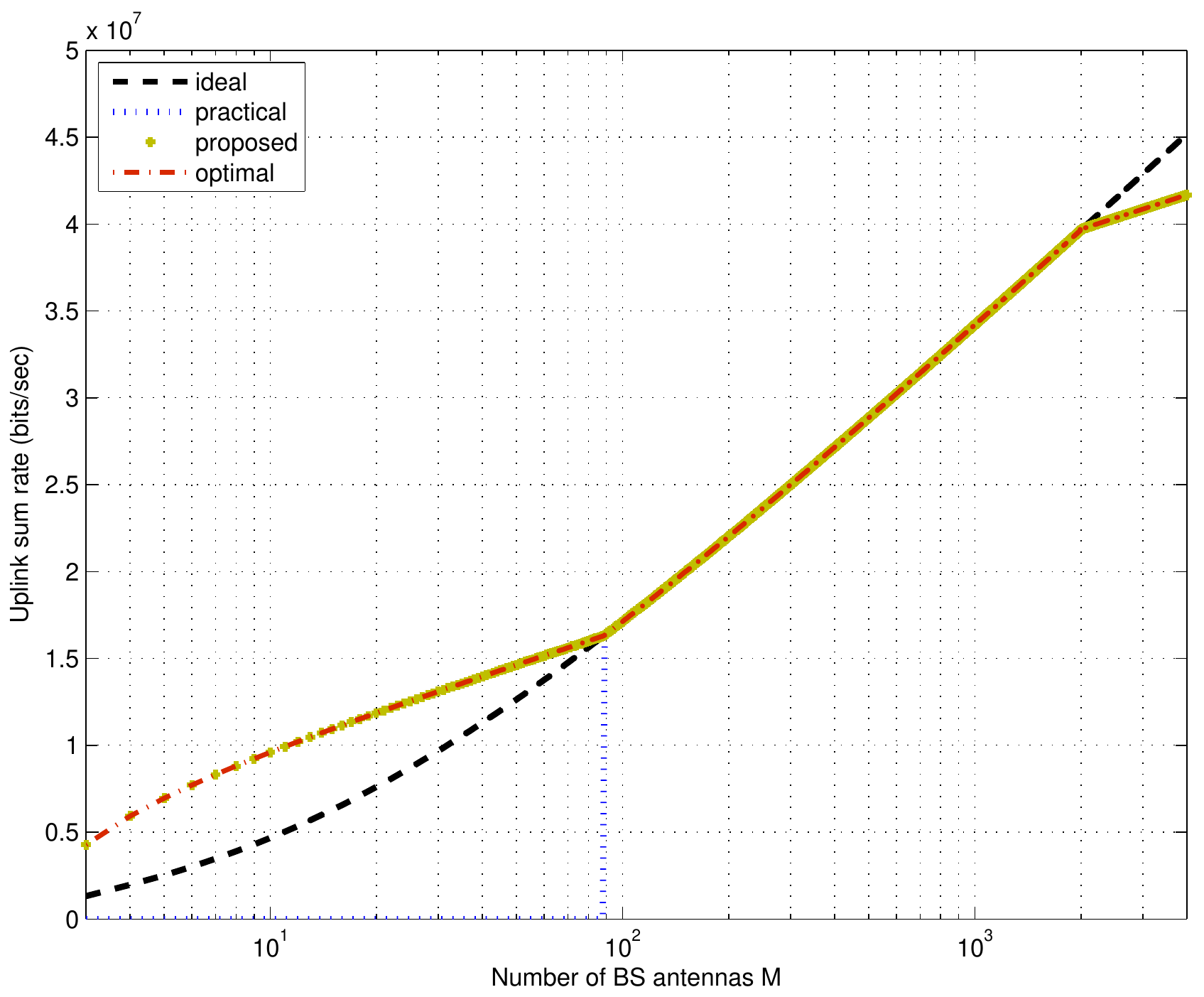}
	}
	\caption{Uplink achievable rate vs. the number of BS antennas.
		The proposed solution yields uplink rate almost similar to that obtained using the EE-optimal solution. The rate increases monotonically with $M$ even in the saturated mode due to improved uplink detection.
		}
	\label{fig:rate}
\end{figure}

\section{Conclusions}\label{secConc}
We optimized the system-level power transfer efficiency and 
energy efficiency of wireless energy and/or information transfer in a multi-user network, where  
a BS equipped with a massive antenna array remotely powers multiple single-antenna energy harvesting users.   
Using a piecewise linear function for modeling the harvester output, we derived the average harvested power at a user in terms of the system parameters. 
We then analyzed the power transfer efficiency of the overall system, while using a scalable power consumption model at the BS. 
We found that the overall PTE may increase or decrease by adding more BS antennas, depending on the system parameters such as the number of users, energy harvester specification, and the transmit/circuit power consumption at the BS.
We also found that it tends to improve by adding more users to the system.
We analytically characterized the PTE-optimal values for the number of BS antennas and users. The results suggest that it is PTE-optimal to operate the system in the massive antenna regime.

We also studied the energy efficiency of the overall system when the users communicate with the BS using the harvested energy.
 We characterized the EE-optimal BS transmit power for energy efficient system operation. 
 The analysis, aided by simulations, revealed several useful insights. 
 While the energy efficiency eventually vanishes as the number of antennas becomes large, results suggest that it is energy efficient to operate the system in the massive MIMO regime. Moreover, increasing the transmit power 
helps improve the energy efficiency as the number of antennas is increased. 

\allowdisplaybreaks
\section*{Appendix}
\subsection{Proof of Lemma \ref{lemma:avgrcvperfectcsi}}
Let us derive the average received energy $\bar{\gamma}_i=\alpha_{\rm{WET}}\,\mathbb{E}\left[|y_i|^2\right]$. Using (\ref{eq:signal}), we express $\mathbb{E}\left[|y_i|^2\right]$ as 
\begin{align}\label{eq:proof Lemma 1}
	\mathbb{E}\left[|y_i|^2\right]&={\zeta_i}\mathbb{E}\left[{\|\bg_i\|}^2|s|^2\right]+
	\sum_{j\neq i}^{K}{\zeta_j}\mathbb{E}\left[\|\bg_i^{\rm{H}}{{\hat{\bg}_j}}\|^2 |s|^2\right]\nonumber\\
&\qquad +
\sum_{j\neq i}^{K}
	\sqrt{\zeta_i\zeta_j}	\mathbb{E}\left[\dot{\bg}_i^{\rm{H}}{\hat{\bg}_j}|s|^2\right]	\nonumber\\
&\qquad+\sum_{u\neq i}^{K}
\sum_{v\neq i,u}^{K}
\sqrt{\zeta_u\zeta_v}	\mathbb{E}\left[\hat{\bg}_v^{\rm{H}}{{\bg}_i}{{\bg}_i}^{\rm{H}}\hat{\bg}_u|s|^2	\right],
\end{align}
where ${\hat{\bg}_j}=\frac{{\bg}_j}{\|{\bg}_j\|}$ and ${\dot{\bg}_j}={\|{\bg}_j\|}{{\bg}_j}~\forall~j$.
The rest of the proof follows from the independence of the random vectors $\{\bg_k\}_{k=1}^{K}$ and by further noting that the entries of $\bg_k$ are independent and identically distributed with mean zero and variance $\beta_k$. Furthermore, $\mathbb{E}\left[|s|^2\right]=p_{\rm{dl}}$ and the transmitted symbol $s$ is independent of $\bg_k$. Specifically, the first term in (\ref{eq:proof Lemma 1}) $\zeta_i\mathbb{E}\left[{\|\bg_i\|}^2|s|^2\right]=\zeta_iM\beta_ip_{\rm{dl}}$. Similarly, the second term $\sum_{j\neq i}^{K}{\zeta_j}\mathbb{E}\left[\|\bg_i^{\rm{H}}{{\hat{\bg}_j}}\|^2 |s|^2\right]=\beta_ip_{\rm{dl}}\sum_{j\neq i}^{K}{\zeta_j}=\beta_ip_{\rm{dl}}\left(1-\zeta_i\right)$ since $\sum_{j=1}^{K}{\zeta_j}=1$. The remaining two terms in (\ref{eq:proof Lemma 1}) are zero because $\mathbb{E}\left[\dot{\bg}_i^{\rm{H}}{\hat{\bg}_j}|s|^2\right]=0$ and $\mathbb{E}\left[\hat{\bg}_v^{\rm{H}}{{\bg}_i}{{\bg}_i}^{\rm{H}}\hat{\bg}_u|s|^2\right]=0$.

\subsection{Proof of Lemma \ref{lem:ls}} 
Let us consider a user $i$ transmitting a training signal over $\tau$ symbols with an average symbol energy $p_{{\rm{Tr}},i}$. Assuming the BS uses MMSE channel estimation, the mean received energy $\bar{\gamma}_i^{\rm{I}}=\alpha_{\rm{WET}}\,\mathbb{E}\left[|y_i|^2\right]$ (where $y_i$ follows from (\ref{eq:signal main})) can be expressed as 
\begin{align}
\bar{\gamma}_i^{\rm{I}}=A_1M\left[1-\frac{M-1}{M}\frac{1}{1+\frac{\beta_i\tau p_{{\rm{Tr}},i}}{\sigma^2}}\right] + A_2.
\end{align}
Here, similar to \cite[Appendix B]{ref2}, we have leveraged the independence of random vectors $\{\bg_i\}_i$, and the fact that the variance of the estimation error is $\frac{\beta_i}{1+\frac{\beta_i\tau p_{{\rm{Tr}},i}}{\sigma^2}}$ (see discussion following (\ref{eq: mmse})). We recall that $\tau p_{{\rm{Tr}},i}=\eta_{\rm{PA}}^{\rm{EH}}\xi_i\bar{\delta}_i^{\rm{I}}S$ since the user employs a fraction $\xi_i$ of the per-frame average harvested energy for uplink transmission. For the non-saturated mode, we apply Lemma \ref{lem:ls harv} and substitute $\bar{\delta}_i^{\rm{I}}=\eta_{\rm{EH}}\bar{\gamma}_i^{\rm{I}}$ to obtain a quadratic equation in $\bar{\gamma}_i^{\rm{I}}$.  The solution of this quadratic equation yields (\ref{eq: rcv ls act}). Similarly, we substitute $\bar{\delta}_i^{\rm{I}}=\frac{\eta_{\rm{EH}}\theta_{\rm{sat}}}{B}$ for the saturated mode to obtain (\ref{eq:rcv ls sat}).

\subsection{Proof of Proposition \ref{lem:pte M}} First, note that the PTE is sub-optimal when $M\notin\left[M_{\rm{act}},M_{\rm{sat}}\right]$, as it is zero for $M<M_{\rm{act}}$
and is upper bounded by ${\rm{PTE}}\left(M_{\rm{sat}},K\right)$ for $M>M_{\rm{sat}}$.
The next step is to solve the linear fractional program $\underset{x}{\max}~f(x)=\underset{x}{\max}~\frac{N_1x+N_2}{D_1x+D_2}$ under the constraint 
$x\in\left[M_{\rm{act}},M_{\rm{sat}}\right]$, where
$N_1=BKA_1$, $N_2=BKA_2$, $D_1=P_{\rm{BS}}+{{\tilde{P}_{\rm{CE}}}}+{\tilde{P}_{\rm{LP}}}$, and $D_2=P_{\rm{TX}}
+P_{\rm{FIX}}$. 
We note that $f(x)$ is quasilinear and monotonic in $x$ \cite[Section 4.3]{boyd2004convex},\cite{charnes1962programming}. 
When $N_1D_2>N_2D_1$ such that
$\frac{\partial}{\partial x}f(x)>0$, $x^*=M_{\rm{sat}}$ as $f(x)$ is an increasing function of $x$. Conversely, when $N_1D_2\leq N_2D_1$, $x^*=M_{\rm{act}}$ maximizes $f(x)$.

\subsection{Proof of Proposition \ref{lem: pte K}} 
The proof follows by noting that the function $f(x)=\frac{N_1x+N_2}{D_1x^2+D_2x+D_3}$ is quasiconcave for $x\in\mathbb{R}$ when the superlevel sets $\mathcal{S}_\nu=\{x:f(x)\geq\nu\}$ are convex for any $\nu\in\mathbb{R}$\cite[Section 3.4]{boyd2004convex}. Using differentiation, we can prove the convexity of the superlevel sets for nonnegative values of $\{N_i\}_{i=1,2}$ and $\{D_i\}_{i=1,2,3}$, where $D_1=\dot{P}_{\rm{CE}}$, $D_2=\dot{P}_{\rm{LP}}$, and $D_3={P_{\rm{TX}}+P_{\rm{FIX}}+MP_{\rm{BS}}}$.
$N_1=\eta_{\rm{EH}}\theta_{\rm{sat}}$ and $N_2=0$ in the saturated mode, whereas $N_1=\eta_{\rm{EH}}P_{\rm{dl}}\beta$ and $N_2=\eta_{\rm{EH}}P_{\rm{dl}}\beta\left(M-1\right)$ in the non-saturated mode.  
Solving $\frac{\partial}{\partial K} f(K)=0$, we obtain the optimal $K$ for the saturated $\left(K^{*}_{\rm{sat}}\right)$ or non-saturated mode $\left(K^{*}_{\rm{act}}\right)$, which follows from (i) the quasiconcavity of $f(K)$ since it is an increasing (or decreasing) function of $K$ for $K<K^*$ (or $K>K^*$) where $K^*$ is the stationary point of $f(K)$; and (ii) because $K \leq K_{\rm{sat}}$ and $K \leq K_{\rm{max}}$ in the saturated or non-saturated mode.
Finally, the condition in (\ref{eq:PTE K condition}) is obtained by comparing the maximal PTE in the saturated and non-saturated modes.

\subsection{Proof of Lemma \ref{lem: optimal power}} The proof follows by casting the EE expression in (\ref{eq:EE}) in the form $f(z)=\frac{g\log(1+bz)}{c+dz+h\log(1+bz)}$, where
the constants $c, h\geq 0$, and $b, d, g >0$. Using the quasiconcavity of the function $f(z)$, it was shown in \cite[Lemma 3]{ref3} that the optimal solution to the problem
$\max\limits_{z>-\frac{1}{b}} f(z)=\frac{g\log(1+bz)}{c+dz+h\log(1+bz)}$
is given by
$z^*=\frac{e^{\rm{W}\left[\frac{bc}{de}-\frac{1}{e}\right]+1}-1}{b}$, which completes the proof.
\balance
\bibliographystyle{ieeetr}
\bibliography{references_wrkshp,All_Ref_Ahmed_June.bib}
\end{document}